\documentclass[twocolumn,10pt,superscriptaddress,amsmath,amssymb,aps,prb,floatfix,]{revtex4-2}

\usepackage{graphicx}
\usepackage{bm}
\usepackage{times}
\usepackage{amsmath,amssymb}
\usepackage{enumerate}
\usepackage{multirow}
\usepackage{xcolor}
\usepackage{comment}
\usepackage{amsthm}

\usepackage[colorlinks=true,linkcolor=blue,citecolor=red, linktocpage=true,breaklinks=true]{hyperref}

\newcommand{\up}{\uparrow}
\newcommand{\eq}{\begin{equation}}
\newcommand{\en}{\end{equation}}
\newcommand{\eqa}{\begin{eqnarray}}
\newcommand{\ena}{\end{eqnarray}}

\newtheorem{theorem}{Theorem}
\newtheorem{lemma}{Lemma}

\begin{document}

\title{Absence of local conserved charges of the Fredkin spin chain and its truncated versions}

\author{Wen-Ming Fan}
\affiliation{School of Physics, Northwest University, Xi'an 710127, China}

\author{Kun \surname{Hao}}
\affiliation{Institute of Modern Physics, Northwest University, Xi'an 710127, China}
\affiliation{Shaanxi Key Laboratory for Theoretical Physics Frontiers, Xi'an 710127, China}
\affiliation{Peng Huanwu Center for Fundamental Theory, Xi'an 710127, China}

\author{Yang-Yang Chen}
\affiliation{Institute of Modern Physics, Northwest University, Xi'an 710127, China}
\affiliation{Shaanxi Key Laboratory for Theoretical Physics Frontiers, Xi'an 710127, China}
\affiliation{Peng Huanwu Center for Fundamental Theory, Xi'an 710127, China}

\author{Kun Zhang}
\email{kunzhang@nwu.edu.cn}
\affiliation{School of Physics, Northwest University, Xi'an 710127, China}
\affiliation{Shaanxi Key Laboratory for Theoretical Physics Frontiers, Xi'an 710127, China}
\affiliation{Peng Huanwu Center for Fundamental Theory, Xi'an 710127, China}

\author{Xiao-Hui Wang}
\email{xhwang@nwu.edu.cn}
\affiliation{School of Physics, Northwest University, Xi'an 710127, China}
\affiliation{Shaanxi Key Laboratory for Theoretical Physics Frontiers, Xi'an 710127, China}
\affiliation{Peng Huanwu Center for Fundamental Theory, Xi'an 710127, China}

\author{Vladimir \surname{Korepin}}
\affiliation{C.N. Yang Institute for Theoretical Physics, Stony Brook University, New York 11794, USA}

\date{\today}

\begin{abstract}

Conservation laws serve as the hallmark of integrability. The absence of conserved charges typically implies that the model is nonintegrable. The recently proposed Fredkin spin chain exhibits rich structures, and its ground state is analytically known. However, whether the Fredkin spin chain is integrable remains an open question. In this work, through rigorous analytical calculations, we demonstrate that the Fredkin spin chain, under both periodic and open boundary conditions, lacks local conserved charges, thereby confirming its nonintegrable nature. Furthermore, we find that when one or a portion of the Hamiltonian terms are removed (referred to as the truncated Fredkin spin chain), local conserved charges are still absent. Our findings suggest that in models involving three-site interactions, integrable models are generally rare.

\end{abstract}

\maketitle

\section{\label{sec:intro} Introduction}

Integrable and solvable models have played pivotal roles in the development of nearly all branches of physics. Unlike classical integrability, which is rigorously defined by the Liouville-Arnold theorem \cite{Arnold1989Mathematical}, quantum integrability lacks a universally accepted definition and is often context-dependent \cite{Weigert1992problem,Caux2010RemarksOT}. Quantum many-body models solved using the algebraic Bethe ansatz and the quantum inverse scattering method \cite{Faddeev1996How,korepin1997quantum}, which are based on the Yang-Baxter equation \cite{yang1967some,yang1968s,baxter1972partition}, can be referred to as Yang-Baxter integrable. This method not only provides the eigenvalues and eigenstates of the system but also ensures that the transfer matrix generates local conserved charges that are linearly independent of the Hamiltonian. The existence of nontrivial conserved charges is analogous to the classical Liouville-Arnold integrability framework.

The presence of local conserved charges in quantum integrable models gives rise to various unconventional behaviors in non-equilibrium dynamics. One notable example is the absence of thermalization \cite{rigol2007relaxation,langen2015experimental,essler2016quench}, and the generalized Gibbs ensemble has been proposed \cite{Vidmar2016GeneralizedGE}. With the advancement of quantum technology, it has been found that local conserved charges can be effectively applied in benchmarking quantum computers \cite{maruyoshi2023conserved} and preparing eigenstates in integrable models \cite{Lutz2025AdiabaticQS}.

Mathematically, there is no general method to determine whether a system is Yang-Baxter integrable, making it seemingly impossible to prove nonintegrability directly. However, the existence of local conserved charges can be computationally verified, providing a practical approach to assess integrability. Initially proposed by Grabowski and Mathieu, the conjecture states that the existence of 3-local conserved charges in a 2-local Hamiltonian model serves as a necessary and sufficient condition for integrability \cite{grabowski1995integrability}. Recently, Shiraishi demonstrated the absence of local conserved charges in the XYZ-h model \cite{shiraishi2019proofXYZh}, showing the feasibility of proving the non-existence of $k$-local conserved charges without any assumptions. Subsequent research has extended this approach, proving that the mixed-field Ising model \cite{chiba2024proofIsing}, models with next-nearest-neighbor interactions \cite{shiraishi2024absenceXIX,shiraishi2025complete}, the PXP model \cite{park2025graph,park2025nonintegrability}, the spin-1 bilinear-biquadratic model \cite{hokkyo2024proof,park2025proofspin1}, the isotropic spin models \cite{Shiraishi2025DichotomyTS}, higher-dimensional Hubbard model \cite{futami2025absence}, and higher-dimensional spin models \cite{shiraishi2024s,chiba2025proofhighdimension,Futami2025AbsenceON} do not possess any unknown integrable structures \cite{yamaguchi2024classification,yamaguchi2024proof}. The same methodology can also be applied to construct a complete set of local conserved charges for quantum integrable models \cite{Nozawa2020ExplicitCO,Fukai2023All,Fukai2024Proof}, highlighting its versatility in both proving nonintegrability and establishing integrability. Note that numerical tests on the existence of local conserved charges are possible, but they are limited to $k$-local conserved charges with small $k$ \cite{Moudgalya2023NumericalMF,Bentsen2019IntegrableAC,Zhan2023LearningCL,Pawowski2025FindingLI}.

Although several models have been proven to lack local conserved charges, the integrability of the recently proposed Fredkin spin chain \cite{salberger2018fredkin}, with three-site nearest-neighbor interactions, has not yet been addressed. The ground state of the Fredkin spin chain has an analytical expression and is the uniform superposition of Dyck states, which exhibits unusually strong entanglement. The unusual entanglement of the ground state has motivated extensive studies of the model and its extensions \cite{salberger2017deformed,udagawa2017finite,zhang2017entropy,padmanabhan2019quantum,voinea2024deformed}.
Numerical studies on the Fredkin spin chain have revealed several signatures typically associated with nonintegrability, including Wigner-Dyson level statistics \cite{causer2024nonthermal}, Hilbert space fragmentation and quantum many-body scars \cite{langlett2021hilbert,causer2024nonthermal}, and behaviors of nonequilibrium dynamics \cite{adhikari2021slow,singh2021subdiffusion}. While these phenomena presents a compelling physical picture of nonintegrability, a rigorous mathematical proof is still necessary.

The motivation for investigating the integrability of the Fredkin spin model is twofold. Firstly, research on the integrability of models with medium-range interactions and partially solvable systems is rare. Our objective is to address this gap and contribute to a more comprehensive understanding of the integrability landscape. Secondly, the integer-spin counterpart of the Fredkin spin model, known as the Motzkin model \cite{Movassagh2014PowerLV}, exhibits unusual integrability properties when certain Hamiltonian terms are removed \cite{Tong2020ShorMovassaghCL,Hao2022ExactSO}. This intriguing observation raises the question of whether a similar integrable structure exists in the Fredkin spin model.

In this work, we systematically search for local conserved charges in the Fredkin spin model and its variants, employing a strategy similar to that of Shiraishi’s work \cite{shiraishi2019proofXYZh}. We rigorously prove that the Fredkin spin model is nonintegrable, namely the absence of local conserved charges, under both periodic and open boundary conditions. The proof is more complex than that for Heisenberg models with nearest-neighbor and next-nearest-neighbor interactions, due to the three-site interactions in the Fredkin spin model. Furthermore, we establish the nonintegrability of variants of the Fredkin spin chain, obtained by selectively removing certain three-site interactions. Our results do not contradict the integrability structures observed in Fredkin-like constrained models, which are essentially different compared to the Fredkin spin model \cite{singh2023fredkin,mccarthy2025subdiffusive}. Our results contribute to a more comprehensive integrability landscape of spin models and highlight the challenges in extending integrability concepts to systems with medium-range interactions \cite{Gombor2021IntegrableSC}.

The paper is organized as follows. In Section \ref{sec:preli}, we introduce the Fredkin spin model, define the notations, and briefly outline the proof strategy. In Section \ref{sec:absence_3_4_5}, we prove the absence of 4- and 5-local conserved charges in the Fredkin spin model. The proofs for the 4- and 5-local cases differ and are extended to the general even and odd cases in Section \ref{sec:absence_general}. In Section \ref{sec:open_boundary_condition}, we prove that the Fredkin spin model with open boundary conditions does not admit boundary conserved charges. The variants of the Fredkin spin model are investigated in Section \ref{sec:simplied_fredkin}, and they are also shown to be nonintegrable. We conclude the manuscript in Section \ref{sec:conclusion}. Three appendices are included, providing detailed calculations relevant to the proofs.

\section{\label{sec:preli} Preliminary}
\subsection{Model}

Our research subject is the $S = 1/2$ Fredkin spin chain, whose Hamiltonian is given by \cite{salberger2018fredkin}
\begin{equation}\label{origin H}
    H = H_{\partial} + \sum_j H_j,
\end{equation}
where $H_{\partial}$ represents the boundary term, and $H_j$ denotes the Hamiltonian density in the bulk. Specifically, the Hamiltonian density is expressed as
\begin{equation}
    H_j = (1\!\!1 - F_{j,j+1,j+2}) + (1\!\!1 - X_{j+2} F_{j+2,j+1,j} X_{j+2}),
\end{equation}
where $X_j$ is the Pauli-$x$ matrix acting on the $j$-th site, and $F_{j,j+1,j+2}$ is the three-qubit Fredkin gate \cite{nielsenQuantumComputationQuantum2010}, defined as
\begin{equation}
    F_{j,k,l} = |\downarrow\rangle_j \langle \downarrow| \otimes 1\!\!1_k \otimes 1\!\!1_l + |\uparrow\rangle_j \langle \uparrow| \otimes \text{SWAP}_{k,l}\,.
\end{equation}
Here $|\uparrow\rangle$ and $|\downarrow\rangle$ are the eigenstates of the Pauli-$z$ matrix, corresponding to spin-up and spin-down states, respectively. The two-qubit SWAP gate is given by
\begin{equation}\label{eq:SWAP_gate}
    \text{SWAP}_{k,l} = \frac{1}{2} (1\!\!1_k \otimes 1\!\!1_l + X_k \otimes X_l + Y_k \otimes Y_l + Z_k \otimes Z_l),
\end{equation}
where $Y$ and $Z$ are the remaining two Pauli matrices. The Fredkin gate $F_{j,k,l}$ acts as the identity operator if the $j$-th spin is in the $|\downarrow\rangle$ state, while it applies the $\text{SWAP}_{k,l}$ gate if the $j$-th spin is in the $|\uparrow\rangle$ state. Therefore, it is also referred to as the controlled-SWAP gate.

We can rewrite the Hamiltonian density using the definition of the SWAP gate in Eq. (\ref{eq:SWAP_gate}). Furthermore, when considering the periodic boundary condition, the Hamiltonian density is equivalent to (up to an overall constant factor)
\begin{equation}
    H_j \propto H_j^\text{(2)} + H_j^\text{(3)},
\end{equation}
where the nearest-neighbor interaction term is given by
\begin{equation}
    H_j^\text{(2)} = -2 \left( X_j X_{j+1} + Y_j Y_{j+1} + Z_j Z_{j+1} \right),
\end{equation}
and the three-local interaction term is expressed as
\begin{multline}\label{H_tilde}
    H_j^\text{(3)} = X_j X_{j+1} Z_{j+2} + Y_j Y_{j+1} Z_{j+2} \\
    - Z_j X_{j+1} X_{j+2} - Z_j Y_{j+1} Y_{j+2}.
\end{multline}
We have omitted the identity operator in the Hamiltonian since it has no contribution on the conserved charges. Here, $H_j^\text{(2)}$ is precisely the Hamiltonian density of the Heisenberg XXX spin chain \cite{Faddeev1996How}. Therefore, the Fredkin spin chain can be interpreted as a ``dressed'' Heisenberg XXX model.

Under the open boundary condition, the summation index $j$ travels from $1$ to $N-2$, where $N$ denotes the size of the model. The original Fredkin spin chain includes an open boundary term given by $H_{\partial} = |\downarrow\rangle_1 \langle \downarrow| + |\uparrow\rangle_N \langle \uparrow|$. The Fredkin spin chain is frustration-free and possesses a unique ground state, which corresponds to a uniform superposition of Dyck paths. The entanglement entropy of the half-chain scales as $\log N$. This clear violation of the area law indicates that the model is gapless, a result that has been rigorously proven in \cite{movassagh2016gap}.

When considering the periodic boundary condition with $H_\partial = 0$ and $N+j = j$, the ground state becomes degenerate. The degeneracy is related to the parity of the length $N$ \cite{Pronko2025SymmetriesOT}.  In addition to the obvious translation symmetry and $U(1)$ symmetry, the Fredkin spin chain possesses additional nonlocal conserved charges \cite{Pronko2025SymmetriesOT}. It relates to the degeneracy of the ground state under the periodic boundary condition. The nonlocal conserved charges are linear combinations of the total spin operator $S^\pm$. Note that the number of these nonlocal conserved charges is always two, which is independent on the total site number $N$, therefore are irrelevant to the integrability.

\subsection{Proof Idea}

First, we denote a Pauli string of length $l$ as
\begin{equation}\label{Ajl}
    \mathbf{A}_j^l=A_j^1A_{j+1}^2\cdots A_{j+l-2}^{l-1}A_{j+l-1}^l,
\end{equation}
where the subscript $j$ indicates that the starting site is $j$. The superscript $l$ denotes the operator string length. The operator $A_{j+k}^{k+1}$, which acts on the site $j+k$, is the $(k+1)$-th component of the string $\mathbf{A}_j^l$. We require that the first and last operators in the string $\mathbf{A}_j^l$, namely $A_j^1$ and $A_{j+l-1}^l$, belong to the set $\{X, Y, Z\}$, while the intermediate operators can include the identity operator, namely $A_{j+k}^{k+1} \in \{I, X, Y, Z\}$ for $1 \leq k \leq l-2$. For simplicity, we also alternatively use a short notation for Pauli strings $\mathbf{A}_j^l = (A^1A^2\cdots A^l)_{j}^l$. For example, $X_jY_{j+1}Z_{j+2}$ can be denoted as $(XYZ)_j^3$, and $X_jX_{j+3}$ as $(XIIX)_j^4$.

If $l\le N/2$, meaning that $\mathbf{A}_j^l$ acts on consecutive sites spanning less than half the chain length, we call $\mathbf{A}_j^l$ as a local charge. The summation of local operators is also called local. For example, the Fredkin spin chain Hamiltonian is a 3-local operator. Conversely, charges acting on more consecutive sites are called nonlocal. For example, $Z_jX_{j+1}X_{j+2}$ is a $3$-local charge, $X_jZ_{j+1}Y_{j+N/2}$ is a nonlocal one. Note that different boundary conditions also have an impact. For example $X_NY_1$ is $2$-local operator under the periodic boundary condition but becomes $N$-local under the open boundary condition. The definitions of ``local'' and ``nonlocal'' appears to be rather contrived. In fact, such definitions works well to distinguish the conserved charges originated from the Hamiltonian or not \cite{grabowski1995integrability}. For example, $H^2$ is a conserved quantity but nonlocal, therefore not related to integrability. 

A general translation invariant $k$-local operator can be expressed as
\begin{equation}\label{Qexpression}
    Q=\sum_{l=1}^{k} \sum_{j=1}^N \sum_{\mathbf{A}_j^l} q_{\mathbf{A}^l}\mathbf{A}_j^l,
\end{equation}
where $q_{\mathbf{A}^l}$ is the coefficient corresponding to each $\mathbf{A}_j^l$. Since the Hamiltonian density of the Fredkin spin chain consists of at most $3$-local operators, the commutator of $Q$ and $H$ is an at most $(k+2)$-local operator, which can also be expressed in a similar way
\begin{equation}\label{[Q,H]}
    [Q,H]=\sum_{l=1}^{k+2} \sum_{j=1}^N \sum_{\mathbf{B}_j^l} r_{\mathbf{B}^l} \mathbf{B}^l_j,
\end{equation}
where $\mathbf{B}^l_j$ denotes the basis operators of Pauli string with length $l$, while $r_{\mathbf{B}^l}$ are the corresponding coefficients, which have linear relations with $q_{\mathbf{A}^l}$. 

If $[Q,H] = 0$, we call $Q$ as a $k$-local conserved charge of $H$. The Hamiltonian of Fredkin spin chain is a 3-local conserved charge, which is trivial for integrability. We look for the $k$-local conserved charges with $4\le k\le N/2$. To prove whether the Fredkin spin chain has nontrivial local conserved charges, we follow Shiraishi’s strategy \cite{shiraishi2019proofXYZh}. We first set $\left[ Q , H \right]=0$, which gives $r_{\mathbf{B}^l}=0$ for all $\mathbf{B}^l_j$. Through the linear relations between the coefficients $r_{\mathbf{B}^l}$ and $q_{\mathbf{A}^l}$, we solve coefficients $q_{\mathbf{A}^l}$. If the only solution is $q_{\mathbf{A}^l}=0$ for all $\mathbf{A}^l$, then the model is absence of nontrivial conserved charge. Conversely, if $q_{\mathbf{A}^l} \ne 0$ gives $[Q,H] = 0$, it indicates the (Yang-Baxter) integrability of the model. 

Solving $[Q,H] = 0$ is a formidable task. Apparently the linear equation of $q_{\mathbf{A}^l}$ is overdetermined. The trick is to track $q_{\mathbf{A}^k}$ from $r_{\mathbf{B}^{k+2}}=0$. If $q_{\mathbf{A}^k}$ can not be determined by $r_{\mathbf{B}^{k+2}}=0$, then we look for the lower order equations $r_{\mathbf{B}^{k+1}}=0$. In most cases, we do not need considering all equations to determine that $q_{\mathbf{A}^k}=0$.

\subsection{Column expression}

A commutator between a Pauli string operator and a Hamiltonian term (also is a Pauli string operator) is also a Pauli operator, for example,
\begin{equation}
    \left[(ZXYZ)_j^4,(XXZ)_{j+3}^3\right]=2i (ZXYYXZ)_j^6.
\end{equation}
Using the column expression applied in Ref. \cite{shiraishi2019proofXYZh}, the above commutation relation can be written as 
\begin{equation}\label{column expression eg}
\begin{matrix}
  & Z_j & X_{j+1} & Y_{j+2} & Z_{j+3} &  & \\
  &  &  &  & X_{j+3} & X_{j+4} &Z_{j+5} \\\hline
  & Z_j & X_{j+1} & Y_{j+2} & Y_{j+3} & X_{j+4} &Z_{j+5}.
\end{matrix}
\end{equation}
The first row corresponds to Pauli string in $Q$, while the second corresponds to an operator in the Hamiltonian. The bottom row shows the commutator result. Each column aligns with the spin sites. The factor $2i$ is omitted as it can be inferred from each commutator. The column expression has the advantage of clearly illustrating which operators are employed to generate other Pauli operators and which do not interact due to their distinct positions.

As another example, the operator $(ZXYYXZ)_j^6$ can be generated by another commutator, namely
\begin{equation}\label{column expression eg 2}
\begin{matrix}
  & &  & Z & Y & X &Z \\
 - & Z & X & X &  &  & \\\hline
 - & Z & X & Y & Y & X &Z.
\end{matrix}    
\end{equation}
Here we have dropped all subscripts of spin site for brevity, and we do the same in the following when there is no ambiguous. For an $N=8$ Fredkin spin chain with the periodic boundary condition, the operator $(ZXYYXZ)_j^6$, as an example of $\mathbf{B}^l_j$ in Eq. (\ref{[Q,H]}), can only be generated from the commutators (\ref{column expression eg}) and (\ref{column expression eg 2}). Thus, we denote that $(ZXYZ)_j^4$ and $(ZYXZ)_{j+2}^4$ form a pair in generating $(ZXYYXZ)_j^6$, or equivalently, $(ZXYZ)_j^4$ and $(ZYXZ)_{j+2}^4$ are called connected.

Furthermore, the two commutators (\ref{column expression eg}) and (\ref{column expression eg 2}) gives a linear relation among coefficients
\begin{equation}
q_{ZXYZ}-q_{ZYXZ}=r_{ZXYYXZ},
\end{equation}
where translation invariance has been applied. The conserved charge requires $r_{ZXYYXZ}=0$, therefore we get a constraint
\begin{equation}
    q_{ZXYZ}=q_{ZYXZ}.
\end{equation}
Similar analysis on other $4$-local operators can lead to a series of linear relations of the coefficients $q_{\mathbf{A}^l}$. Their solutions allow us to confirm the existence or absence of conserved charges. We employ this approach systematically throughout the subsequent analysis.

\section{\label{sec:absence_3_4_5} Absence of 4- or 5-local conserved charges}

As noted earlier, the Hamiltonian of Fredkin spin chain is a trivial 3-local conserved charge. It is well-known that Heisenberg XXX spin chain is integrable and therefore possess a tower of local conserved charges \cite{Grabowski1994QuantumIO}. We have verified that the 3-local conserved charge of Heisenberg XXX spin chain does not commute with the three-site interaction $H_j^{(3)}$ (\ref{H_tilde}) of the Fredkin spin chain, implying that the Fredkin spin chain Hamiltonian is the only 3-local conserved charge. In this section, as a warm-up for the general case, we demonstrate that
\begin{theorem}
\label{theo:local_4_5_charge}
   The Fredkin spin chain does not possess any 4- or 5-local conserved charges under periodic boundary conditions.
\end{theorem}
We examine the $4$- and $5$-local charges in Secs. \ref{sec:absence_4} and \ref{sec:absence_5}, respectively. These cases are representative examples since the treatment of even and odd $k$-local charges differs. We present the general odd and even $k$ in Sec. \ref{sec:absence_general}.

\subsection{\label{sec:absence_4}Absence of 4-local conserved charges}     

The commutator of a 4-local operator and the Fredkin spin chain Hamiltonian can generate 6-local operators. In other words, the 6-local operator is generated from the commutator where the 4-local operator and the 3-local operator in the Hamiltonian only have one site overlap. We list all the generated 6-local operators in the Appendix \ref{App_sec: all_6-local}. See Tables \ref{table:6-local 1}-\ref{table:6-local 9}. We find that all commutators can be categorized into three cases.
\begin{itemize}
    \item\label{item:case1} Case 1: The $4$-local charge cannot form a pair in generating any possible $6$-local operators. It corresponds to the operators with rows that all $6$-local operators are uncolored in Tables \ref{table:6-local 1}-\ref{table:6-local 9}. By the requirement of conserved charge condition $[Q,H] = 0$, the coefficients of these operators have to be zero. An example is $ZIIZ$, which can generate four different $6$-local operators, namely
    \begin{align}
    \begin{split}
    \begin{matrix}
      &  &  & Z & I & I &Z \\
      -& Z & X & X &  &  & \\\hline
     - & Z & X & Y & I & I &Z,
    \end{matrix}\qquad
    \begin{matrix}
      &  &  & Z & I & I &Z \\
     - & Z & Y & Y &  &  & \\\hline
      & Z & Y & X & I & I &Z,
    \end{matrix}\\
    \begin{matrix}
      & Z & I & I & Z &  & \\
      &  &  &  & X & X &Z \\\hline
      & Z & I & I & Y & X &Z,
    \end{matrix}\qquad
    \begin{matrix}
      & Z & I & I & Z &  & \\
      &  &  &  & Y & Y &Z \\\hline
      -& Z & I & I & X & Y &Z.
    \end{matrix}
    \end{split}
    \end{align}
    However, the above generated four $6$-local operators cannot be generated by other commutators, which implies that $ZIIZ$ has a zero coefficient in the conserved charges (if it exists).. 

    \item\label{item:case2} Case 2: The $4$-local charge that can form pairs when commutating with specific Hamiltonian terms but fails to pair with other Hamiltonian terms. Charges in this case correspond to partial colored rows in Tables \ref{table:6-local 1}-\ref{table:6-local 9}. The conserved charge condition $[Q,H] = 0$ also requires the coefficients of these $4$-local charges are zero. An example in this case is the operator $ZZXZ$, which gives the following commutators
    \begin{align}
    \begin{split}
    \begin{matrix}
      &  &  & Z & Z & X &Z \\
      -& Z & X & X &  &  & \\\hline
     - & Z & X & Y & Z & X &Z,
    \end{matrix}\qquad
    \begin{matrix}
      &  &  & Z & Z & X &Z \\
     - & Z & Y & Y &  &  & \\\hline
      & Z & Y & X & Z & X &Z,
    \end{matrix}\\
    \begin{matrix}
      & Z & Z & X & Z &  & \\
      &  &  &  & X & X &Z \\\hline
      & Z & Z & X & Y & X &Z,
    \end{matrix}\qquad
    \begin{matrix}
      & Z & Z & X & Z &  & \\
      &  &  &  & Y & Y &Z \\\hline
      -& Z & Z & X & X & Y &Z.
    \end{matrix}
    \end{split}
    \end{align}
    The first two $6$-local operators, namely $ZXYZXZ$ and $ZYXZXZ$, can also be generated by $ZXYY$ and $ZYXY$, respectively. But the remaining two $6$-local operators can only be generated by $ZZXZ$, therefore we get the coefficient $q_{ZZXZ}=0$ from $[Q,H] = 0$.
    
    \item\label{item:case3} Case 3: The $4$-local charge can form pairs in generating $6$-local operators from all terms in the Hamiltonian. These charges correspond to rows where all $6$-local operators are colored in Tables \ref{table:6-local 1}-\ref{table:6-local 9}.  It leads to a set linear equations of the coefficients, and we cannot conclude that these $4$-local charges have zero coefficient directly.
\end{itemize}

Next, we focus on Case 3 and prove that charges in this case have zero coefficients. All $4$-local charges in Case 3 can be grouped into two sets. The first set in Case 3 has the operators
\begin{align}\label{4-local case3 1}
  \begin{split}
    XXXX,\; XXYY,&\; XXYZ,\; ZXYY,\\
    YYYY,\; YYXX,&\; YYXZ,\; ZYXX.
  \end{split}
\end{align}
As an example, we notice that $XXXX$ form a pair with $ZYZY$ in generating $ZYZXXX$, given by
\begin{equation}
\begin{matrix}
      &  &  & X & X & X &X \\
      -& Z & Y & Y &  &  &\\\hline
      -& Z & Y & Z & X & X &X,
\end{matrix}\qquad
\begin{matrix}
      & Z & Y & Z & Y &  & \\
      &  &  & - & Z & X &X \\\hline
     - & Z & Y & Z & X & X &X.  
\end{matrix}
\end{equation}
The two commutators yield the constraint
\begin{equation}
    -q_{XXXX}-q_{ZYZY}=0.
\end{equation}
Since $ZYZY$ belongs to the Case 2 with $q_{ZYZY}=0$, we immediately obtain $q_{XXXX} = 0$. Applying the similar analyze, we conclude that all $4$-local charges in Eq. (\ref{4-local case3 1}) have zero coefficients.

The second set in Case 3 include the operators
\begin{align}\label{4-local case3 2}
    ZXYZ,&\; ZYXZ.
\end{align} 
The above two operators always form pairs with each other when generating $6$-local operators. Their coefficients yield four linear equations
\begin{subequations}\label{four 4-local relations}
\begin{align}
 -&q_{ZXYZ}-q_{ZXYZ}=0,\\
 &q_{ZYXZ}+q_{ZYXZ}=0,\\
&q_{ZXYZ}-q_{ZYXZ}=0,\\
&q_{ZXYZ}-q_{ZYXZ}=0.   
\end{align}    
\end{subequations}
The last two equations are trivial. The first two gives the solution $q_{ZXYZ} = q_{ZYXZ} = 0$. The above analyze demonstrate that the Fredkin spin chain has no $4$-local conserved charges in the periodic boundary condition.

\subsection{\label{sec:absence_5}Absence of \texorpdfstring{$5$}{5}-local conserved charges}

First, we consider a $5$-local charge with the rightmost $X$ operator, namely $A_j^1A_{j+1}^2A_{j+2}^3\cdots X$ (similar analyze can be applied if the rightmost is $Y$ or $Z$). Only the rightmost operator $X$ matters if we consider the generated $7$-local operator from the commutator with the 3-local Hamiltonian. For example, we have
\begin{align}
\begin{matrix}\label{derive_5-local_sequence_1}
  & A_j^1 & A_{j+1}^2 & A_{j+2}^3 & \cdots & X_{j+4} &  & \\
  &  &  &  &  & Y_{j+4} & Y_{j+5} & Z_{j+6}\\\hline
  & A_j^1 & A_{j+1}^2 & A_{j+2}^3 & \cdots & Z_{j+4} & Y_{j+5} &Z_{j+6}.
\end{matrix}
\end{align}
The above $7$-local operator can also be generated by
\begin{align}\label{derive_5-local_sequence_2}
\begin{matrix}
  &  &  & \widetilde{A}_{j+2}^3 & \cdots & Z_{j+4} & Y_{j+5} &Z_{j+6}\\
  & X_j & X_{j+1} & Z_{j+2} &  &  &  & \\\hline
  & A_j^1 & A_{j+1}^2 & A_{j+2}^3 & \cdots & Z_{j+4} & Y_{j+5} &Z_{j+6}.
\end{matrix}
\end{align}
Here we take $XXZ$ in the above commutator as example. Note that $\widetilde{A}_{j+2}^3 = Y$ if $A_{j+2}^3 = X$ and $\widetilde{A}_{j+2}^3 = -X$ if $A_{j+2}^3 = Y$. It could be any term in $H_j^{(3)}$ defined in Eq. (\ref{H_tilde}). The commutator (\ref{derive_5-local_sequence_2}) tells us that $A_j^1A_{j+1}^2$ must take $XX$, while $A_{j+2}^3$ should be different from $Z$. In other words, only the operators with the left boundary sequences $XXX\cdots$ and $XXY\cdots$ can form pairs with other five-local charges by the aforementioned approach (\ref{derive_5-local_sequence_1}) and (\ref{derive_5-local_sequence_2}), therefore to be potential conserved charges. By analyzing all the Hamiltonian terms in Eq. (\ref{H_tilde}), we conclude that the operators can form pair only if they have the left sequences
\begin{align}\label{3operator sequence left}
\begin{split}
    &XXX\cdots,\quad XXY\cdots, \quad YYX\cdots, \quad YYY\cdots,\\
    &ZXY\cdots, \:\quad ZXZ\cdots, \:\quad  ZYX\cdots, \quad  ZYZ\cdots.
\end{split}
\end{align}
By analyzing the rightmost three operators similarly, we obtain that operators form pair if the rightmost three operators are
\begin{align}\label{3operator sequence right}
\begin{split}
   &\cdots XXX, \quad \cdots YXX, \quad \cdots XYY, \quad \cdots YYY,\\
    &\cdots YXZ, \:\quad \cdots ZXZ, \:\quad \cdots XYZ, \quad \cdots ZYZ.     
\end{split}
\end{align}
Combining the permitted left and right sequences, we get a set of 5-local operators which can form pair when commutating any terms in the Hamiltonian $H_j^{(3)}$ (\ref{H_tilde}). It is similar to the case 3 defined in Sec. \ref{sec:absence_4}. Charges with other boundary sequences have zero coefficients.

The set of charges with the permitted boundary sequences can also be gourped into two sets. The first set is similar to charges in Eq. (\ref{4-local case3 1}), which pair with charges in Case 2, therefore have vanishing coefficients. The left charges are
\begin{align}\label{5-local_nonzero}
\begin{split}
    XXXYZ, \quad YYYXZ, &\quad ZXYYY, \quad ZYXXX, \\
    XXYXZ, \quad YYXYZ, &\quad ZXYXX, \quad ZYXYY, \\
    ZXZYZ, &\quad ZYZXZ.
\end{split}
\end{align}
The ten charges form pairs with each other, which require further analyze. 

The strategy is to exam additional operators generated by the above ten $5$-local charges. We first notice that $ZYZXZ$ and $XXYXZ$ can form pair in generating $ZYZXYXZ$, which gives
\begin{equation}\label{5-5 relation}
    q_{ZYZXZ}-q_{XXYXZ}=0.
\end{equation}
Next, we find that all commutators generating the $5$-local operator $XIZXZ$ are
\begin{align}
    \begin{matrix}
        &X&I&Z&Y&\\
        -2&&&&Z&Z\\\hline
        -2&X&I&Z&X&Z,
    \end{matrix}
\end{align}
\begin{align}
    \begin{matrix}
        &X&X&Y&X&Z\\
        -2&&X&X&&\\\hline
        2&X&I&Z&X&Z,
    \end{matrix}\qquad
    \begin{matrix}
        &Z&Y&Z&X&Z\\
        -2&Y&Y&&&\\\hline
        2&X&I&Z&X&Z.
    \end{matrix}
\end{align}
The $4$-local charge $XIZY$ has zero coefficient because it cannot form pair with other $4$-local charges as well as charges in Eq. (\ref{5-local_nonzero}) in generating any $6$-local operator. The remaining two operators give the relation
\begin{align}
    2q_{XXYXZ}+2q_{ZYZXZ}=0.
\end{align}
Combining this equation with Eq. (\ref{5-5 relation}) gives the solution $q_{XXYXZ}=q_{ZYZXZ}=0$. Since the ten charges in Eq. (\ref{5-local_nonzero}) are connected, therefore all have a zero solution. Then we can confirm that the Fredkin spin chain has no $5$-local conserved charges.

\section{\label{sec:absence_general} Absence of general \texorpdfstring{$k$}{k}-local conserved charges}

In this section, we investigate whether general $k$-local conserved charges exist. We find that the Theorem \ref{theo:local_4_5_charge} can be generalized to
\begin{theorem}
    The Fredkin spin chain has no $k$-local conserved charges with $4\le k\le N/2$ under the periodic boundary condition.
\end{theorem}
The proof for the even $k$ (Sec. \ref{sec:general_even_k}) and odd $k$ (Sec. \ref{sec:general_odd_k}) differs significantly, which are extensions of the methods developed for $4$-local and $5$-local charges respectively.

\subsection{\label{sec:derive_sequence} Boundary sequences of paired operators}

We start from the derivation of boundary sequences of connected operators, which generalizes the operators in Eqs. (\ref{3operator sequence left}) and (\ref{3operator sequence right}) when considering the 5-local conserved charges. For convenience, we focus on the left boundary sequences. Extending all results to the right boundary is straightforward by symmetry.   

First let us consider a $k$-local operator $(XXAA\cdots AAX)_j^k$ with arbitrary operators $A$ in the middle. The rightmost operator can be arbitrary, here we take $X$. We call it the {\itshape charge 1}. {\itshape Charge 1} has the following commutation relation (only keeping necessary subscripts)
\begin{equation}
\begin{matrix}
\label{eq:check}
  &X_j& X_{j+1} & \cdots & \cdots & X_{j+k-1} &  &  & \\
 & &  &  &  -& Z_{j+k-1} & X & X & \\\hline
 &X_j& X_{j+1}& \cdots & \cdots &  Y & X & X_{j+k+1} &.
\end{matrix}
\end{equation}

Next we try to search for another $k$-local charge that can generate the above $(k+2)$-local operator. We notice that the two leftmost operators of the above $(k+2)$-local operator is $XX$, which is consistent with the Hamiltonian term $XXZ$. This implies the $(k+2)$-local operator can also be generated by the commutator of $XXZ$ and a $k$-local charge, i.e.
\begin{equation}
\begin{matrix}
  &  &  & A_{j+2}^1 & A_{j+3}^2& \cdots &  Y & X & X_{j+k+1}  \\
  &X_j& X_{j+1}& Z_{j+2} && & & & \\\hline
&X_j& X_{j+1}& \tilde{A}_{j+2}^3 & \tilde A_{j+3}^4&\cdots & Y & X & X_{j+k+1}.
\end{matrix}
\end{equation}
We denote $A_{j+2}^1A_{j+3}^2\cdots Y_{j+k-1}X_{j+k}X_{j+k+1}$ as the {\itshape charge 2}. For {\itshape charge 2}, the boundary sequence $A_{j+2}^1 A_{j+3}^2$ must be $X_{j+2}X_{j+3}$ or $Y_{j+2}Y_{j+3}$, otherwise it will generate operators which can not form pairs. Correspondingly, the generated operator $\tilde{A}_{j+2}^3\tilde A_{j+3}^4$ equals to $X_{j+2}Y_{j+3}$ or $Y_{j+2}X_{j+3}$. It also sets the {\itshape charge 1} at sites $j+2$ and $j+3$ as $X_{j+2}Y_{j+3}$ or $Y_{j+2}X_{j+3}$. With the above constrains, the operator $(XX\cdots YXX)_j^{k+2}$ can only be generated by the above two commutators.

Applying the same analyze to {\itshape charge 2}, we confirm possible operators at sites $j+4$ and $j+5$ are $X_{j+4}Y_{j+5}$ or $Y_{j+4}X_{j+5}$ whether {\itshape charge 2} takes the form $(XX\cdots ZXX)_{j+2}^k$ or $(YY\cdots ZXX)_{j+2}^k$. Iterating this procedure generates the permitted boundary sequences. Define $\mathcal{\overline{A}}_j =X_{j}Y_{j+1}$ or $Y_jX_{j+1}$. Then we find that the operator $XX\cdots$ can form pair if it has the form
\begin{equation}
    X_jX_{j+1}\;\mathcal{\overline{A}}_{j+2}\; \mathcal{\overline{A}}_{j+4}\;\mathcal{\overline{A}}_{j+6}\cdots.
\end{equation}

A comprehensive analysis on all operators in Eq. (\ref{3operator sequence left}) gives the $k$-local charges which can form pairs. The representative examples with valid left boundary sequences are
\begin{subequations}\label{nonzero sequence}
\begin{align}
\label{nonzero sequence 1}
    X_jX_{j+1}\; \mathcal{\overline{A}}_{j+2}\; \mathcal{\overline{A}}_{j+4}\;\mathcal{\overline{A}}_{j+6}\cdots,\\
\label{nonzero sequence 2}
    Y_jY_{j+1}\;\mathcal{\overline{A}}_{j+2}\; \mathcal{\overline{A}}_{j+4}\;\mathcal{\overline{A}}_{j+6}\cdots,\\
    Z_jX_{j+1}Z_{j+2}Y_{j+3}\; \mathcal{\overline{A}}_{j+4}\;\mathcal{\overline{A}}_{j+6}\cdots,\\
    Z_jY_{j+1}Z_{j+2}X_{j+3}\; \mathcal{\overline{A}}_{j+4}\;\mathcal{\overline{A}}_{j+6}\cdots,
\end{align}
\end{subequations}
which have even number of boundary operators. The operator 
\begin{align}\label{nonzero sequenceZ}
    Z_j\;\; \mathcal{\overline{A}}_{j+1}\;\mathcal{\overline{A}}_{j+3}\;\mathcal{\overline{A}}_{j+5}\cdots
\end{align}
with odd number of boundary operator is also permitted. See Appendix \ref{App_sec:nonzero_sequence} and Eq. (\ref{more_valid_sequences}) for more valid sequences. The corresponding valid right boundary sequences can be obtained by reflecting the above left boundary sequences. According to our analysis, a charge may have a nonzero coefficient only if its left side operators takes one of above valid sequences, while right side sequence takes one of their inversion.

Before further analyze on the coefficients of above operators, we have the following two lemmas. 
\begin{lemma}\label{lemma:connected}
Operators with the left and right boundary sequences given by Eqs. (\ref{nonzero sequence}) are all connected.
\end{lemma}
\begin{proof}
We take $(XXXYXY\cdots X)_j^k$ and $(YYXYXY\cdots X)_j^k$ as an example. The operators $(XXXYXY\cdots X)_j^k$ and $(YYXYXY\cdots X)_j^k$ belong to operators in Eqs. (\ref{nonzero sequence 1}) and (\ref{nonzero sequence 2}) respectively. First, we find that $(XXXY\cdots X)_j^k$ forms pair with $(YY\cdots ZYZ)_{j+2}^j$, given by
\begin{subequations}\label{proof_connected_1}
\begin{equation}\label{proof_connected_1a}
\begin{matrix}
  & X_j & X & X & Y & \cdots & X_{j+k-1} &  & \\
  &  &  &  &  &  & Y_{j+k-1} & Y & Z\\\hline
  & X_j & X & X & Y & \cdots & Z & Y &Z_{j+k+1},
\end{matrix}
\end{equation}
\begin{equation}\label{proof_connected_1b}
\begin{matrix}
  &  &  & Y_{j+2} & Y & \cdots & Z & Y & Z_{j+k+1}  \\
  & X_j & X & Z &  &  &  &  &  \\\hline
  & X_j & X & X & Y & \cdots & Z & Y & Z_{j+k+1}.
\end{matrix}
\end{equation}    
\end{subequations}
Second, $(YY\cdots ZYZ)_{j+2}^k$ forms pair with $(YYXY\cdots X)_j^k$ from the commutators
\begin{subequations}\label{proof_connected_2}
\begin{equation}\label{proof_connected_2a}
\begin{matrix}
 &  &  & Y_{j+2} & Y & \cdots & Z & Y & Z_{j+k+1}  \\
  & Y_j & Y & Z &  &  &  &  &   \\\hline
  & Y_j & Y & X & Y & \cdots & Z & Y & Z_{j+k+1}, 
\end{matrix}
\end{equation}
\begin{equation}\label{proof_connected_2b}
\begin{matrix}
  & Y_j & Y & X & Y & \cdots & X_{j+k-1} &  &  \\
  &  &  &  &  &  & Y_{j+k-1} & Y & Z  \\\hline
  & Y_j & Y & X & Y & \cdots & Z & Y & Z_{j+k+1}. 
\end{matrix}
\end{equation}  
\end{subequations}
Then we conclude that $(XXXY\cdots X)_j^k$ and $(YYXY\cdots X)_j^k$ are connected. Other operators in Eqs. (\ref{nonzero sequence}) have similar structures, therefore are all connected.
\end{proof}

From Lemma \ref{lemma:connected}, we see that either $XY$ or $YX$ operator in the middle are connected. Formally, we have
\begin{lemma}\label{lemma:change_A}
For operators in Eqs. (\ref{nonzero sequence}), (\ref{nonzero sequenceZ}) and (\ref{more_valid_sequences}), replacing $\mathcal{\overline{A}}_{j} =X_jY_{j+1}$ to $Y_{j}X_{j+1}$ or $\mathcal{\overline{A}}_{j} =Y_jX_{j+1}$ to $X_{j}Y_{j+1}$ preserves connectivity between the corresponding two operators.
\end{lemma}
Lemmas \ref{lemma:connected} and \ref{lemma:change_A} imply that determining any coefficient of operators in Eqs. (\ref{nonzero sequence}), (\ref{nonzero sequenceZ}) and (\ref{more_valid_sequences}) can set the coefficients of the rest operators. 

\subsection{\label{sec:general_even_k}Absence of \texorpdfstring{$k$}{k}-local conserved charges for even \texorpdfstring{$k$}{k}}

\begin{table}[!ht]
    \centering
    \caption{The structure of charges with even length. The charges have valid boundary sequences from Eqs. (\ref{nonzero sequence}) and (\ref{nonzero sequenceZ}) with $A\in \{X,Y\}$. Here ``even" means that the boundary sequences are from Eq. (\ref{nonzero sequence}), and ``odd" is from Eq. (\ref{nonzero sequenceZ}).}
    \begin{tabular}{c|cc|cc}
    \hline
       ~Case~ & ~Left boundary~ & ~Left part~ & ~Right part~ & ~Right boundary~ \\ \hline
       1& even & $\mathcal{\overline{A}}\cdots\mathcal{\overline{A}}$ & $\mathcal{\overline{A}}\cdots\mathcal{\overline{A}}$ & even \\ 
       2& even & $\mathcal{\overline{A}}\cdots\mathcal{\overline{A}}$ & $A\;\mathcal{\overline{A}}\cdots\mathcal{\overline{A}}$ & odd \\ 
       3& even &$\mathcal{\overline{A}}\cdots\mathcal{\overline{A}}\;A$~ & ~$A\;\mathcal{\overline{A}}\cdots\mathcal{\overline{A}}$ & even \\ \hline
       4& odd & $\mathcal{\overline{A}}\cdots\mathcal{\overline{A}}$ & $\mathcal{\overline{A}}\cdots\mathcal{\overline{A}}$ & odd \\ 
       5& odd & $\mathcal{\overline{A}}\cdots\mathcal{\overline{A}}\;A$~ & ~$A\;\mathcal{\overline{A}}\cdots\mathcal{\overline{A}}$ & odd \\ \hline
    \end{tabular}
    \label{table: even structure}
\end{table}

Operators with the boundary sequences given by Eqs. (\ref{nonzero sequence}) and (\ref{nonzero sequenceZ}) are potential conserved charges. The boundary sequences in Eq. (\ref{nonzero sequence}) are even, while being odd if it is given by Eq. (\ref{nonzero sequenceZ}). There are multiple ways for combining together the permitted left and right boundary operators, which give a $k$-local charge. First, we consider the $k$-local charge with even $k$. All possible structures are classified in Table \ref{table: even structure}. We have omitted the cases with the reflection symmetry. For example, if the left boundary sequence is given by Eq. (\ref{nonzero sequenceZ}), the right boundary sequence is either from Eq. (\ref{nonzero sequence}) with a middle isolated operator denoted as $A$, or also from Eq. (\ref{nonzero sequenceZ}). We now demonstrate that Cases 2, 3 and 5 either are equivalent to Case 1 or 4, or have zero coefficient in the conserved charges.

\begin{itemize}
    \item Case 2: We take the operator
    \begin{equation}
        (XX\mathcal{\overline{A}}_{j+2}\cdots | A_{j+m}\mathcal{\overline{A}}_{j+m+1}\cdots Z)_j^k,\nonumber
    \end{equation}
    as en example, where $A_{j+m}\in\{X,Y\}$ and $m$ is even. Here ``$|$" separates the left and right boundary sequences. If $A_{j+m} = X_{j+m}$ and $\mathcal{\overline{A}}_{j+m+1}=X_{j+m+1}Y_{j+m+2}$, we have the sequence $X_{j+m}X_{j+m+1}$ in the operator, which is not permitted. In other words, operators with the sequence $X_{j+m}X_{j+m+1}$ would pair to a zero coefficient charge or generate operators that can not be paired. In either case, it has a zero coefficient. Suppose $A_{j+m} = X_{j+m}$ and $\mathcal{\overline{A}}_{j+m+1}=Y_{j+m+1}X_{j+m+2}$. According to Lemma \ref{lemma:change_A}, we can swap $Y_{j+m+1}X_{j+m+2}$ to $X_{j+m+1}Y_{j+m+2}$, then an illegal sequence $X_{j+m}X_{j+m+1}$ exists. Similarly, if $A_{j+m} = Y_{j+m}$, whether $\mathcal{\overline{A}}_{j+m+1}=Y_{j+m+1}X_{j+m+2}$ or $\mathcal{\overline{A}}_{j+m+1}=X_{j+m+1}Y_{j+m+2}$, we obtain the illegal sequence $\mathcal{\overline{A}}_{j+m}=Y_{j+m}Y_{j+m+1}$. 
    
    \item Case 3: If the two ``$A$" on both sides construct $XY$ or $YX$, it is equivalent to the Case 1. If they construct $XX$ and $YY$, then it is an illegal sequence resulting a zero coefficient.
    
    \item Case 5: Similar to Case 3, if the two ``$A$" on both sides construct $XY$ or $YX$, it is equivalent to Case 4. If they construct $XX$ or $YY$, the operator is illegal. 
\end{itemize}

Essentially, two charges are connected can be regard as changing the rightmost (leftmost) operator and adding two operators on the right (left), while removing the leftmost (rightmost) two operators and changing the third (third from last) operator. For example, the charge $(XXXY\cdots X)_j^k$ in Eq. (\ref{proof_connected_1a}) connected with the charge $(YY\cdots ZYZ)_{j+2}^k$ in Eq (\ref{proof_connected_1b}) can be regarded as performing the following operations on the charge $(XXXY\cdots X)_j^k$: changing the last operator from $X$ to $Z$ and adding $YZ$ on the right (these two operations are caused by commutating with the Hamiltonian term $YYZ$), while removing the first two operators $XX$ and changing the third operator from $X$ to $Y$ (these two operations are caused by commutating with the Hamiltonian term $XXZ$).

Next, we analyze the charges classified in Case 1. Take the example
\begin{equation}\label{even case 1 example}
    (ZXZY\mathcal{\overline{A}}_{j+4}\cdots \mathcal{\overline{A}}_{j+k-4}XX)_j^k.
\end{equation}
Applying the transformation of operator connections described above, we obtain a charge $(AA\cdots ZXZA)_{j-k+4}^k$ with $A\in \{X, Y, Z\}$. It has an illegal right boundary sequence due to the rightmost $Z$. Since all charges in Case 1 are connected due to Lemma \ref{lemma:connected}, we conclude that their coefficients are zero. 

The charges in Case 4 takes the form
\begin{equation}
Z_j\mathcal{\overline{A}}_{j+1}\mathcal{\overline{A}}_{j+3}\cdots\cdots\mathcal{\overline{A}}_{j+k-5}\mathcal{\overline{A}}_{j+k-3}Z_{j+k-1},
\end{equation}
where $\mathcal{\overline{A}}_k=X_kY_{k+1}$ or $\mathcal{\overline{A}}_k=Y_{k}X_{k+1}$. The charge in the above structure exhibit self-pairing behavior. For example, we have
\begin{align}
\begin{matrix}
  & Z_j & X & Y & \cdots & X & Y & Z_{j+k-1} &  & \\
  &  &  &  &  &  &  & Y_{j+k-1} & Y & Z\\\hline
 - & Z_j & X & Y & \cdots & X & Y & X & Y & Z_{j+k+1},
\end{matrix}
\end{align}
\begin{align}
\begin{matrix}
  &  &  & Z_{j+2} & X & Y & \cdots & X & Y &Z_{j+k+1} \\
 - & Z_j & X & X &  &  &  &  &  & \\\hline
- & Z_j & X & Y & X & Y & \cdots & X & Y & Z_{j+k+1}.
\end{matrix}
\end{align}
The two commutators give the coefficient relation
\begin{equation}
    -q_{(ZXY\cdots XYZ)_j^k}-q_{(ZXY\cdots XYZ)_{j+2}^k}=0,
\end{equation}
which is similar to Eq. (\ref{four 4-local relations}). Due to the translation invariance, the above relation gives the solution $q_{ZXY\cdots XYZ}=0$. We thus confirm that all charges in Case 4 have zero coefficients. Combined with the absence of conserved charges in Case 1, we conclude that the Fredkin spin chain admits no $k$-local conserved charge for even $k$. 

\subsection{\label{sec:general_odd_k}Absence of \texorpdfstring{$k$}{k}-local conserved charges for odd \texorpdfstring{$k$}{k}}

We analyze the $k$-local charge with odd $k$ in this subsection. All possible charge structures with valid boundary sequences given by Eqs. (\ref{nonzero sequence}) and (\ref{nonzero sequenceZ}) are classified in Table \ref{table: odd structure}. We have omitted the cases obtained from the reflection symmetry. 

\begin{table}[!ht]
    \centering
    \caption{The structure of charges with odd length. The charges have the valid boundary sequences from Eqs. (\ref{nonzero sequence}) and (\ref{nonzero sequenceZ}) with $A\in \{X,Y\}$.}
    \begin{tabular}{c|cc|cc}
    \hline
       ~Case~ & ~Left boundary~ & ~Left part~ & ~Right part~ & ~Right boundary~ \\ \hline
       1& even & $\mathcal{\overline{A}}\cdots\mathcal{\overline{A}}$ & $\mathcal{\overline{A}}\cdots\mathcal{\overline{A}}$ & odd \\ 
       2& even & $\mathcal{\overline{A}}\cdots\mathcal{\overline{A}}\;A$ & $\mathcal{\overline{A}}\cdots\mathcal{\overline{A}}$ & even \\ 
       3& even &$\mathcal{\overline{A}}\cdots\mathcal{\overline{A}}\;A$~ & ~$A\;\mathcal{\overline{A}}\cdots\mathcal{\overline{A}}$ & odd \\ 
       4& odd & $\mathcal{\overline{A}}\cdots\mathcal{\overline{A}}\;A$~ & ~$\mathcal{\overline{A}}\cdots\mathcal{\overline{A}}$ & odd \\ \hline
    \end{tabular}
    \label{table: odd structure}
\end{table}

Similar to the even $k$ cases, the Case 2-4 either are equivalent to Case 1 or contain illegal sequences. For Case 2, an example is 
\begin{equation}
    (XX\cdots \mathcal{\overline{A}}_{j+m-2}\;|\;A_{j+m}\;X_{j+m+1}Y_{j+m+2}\cdots XX )_j^k,
\end{equation}
 where $A_{j+m}\in\{X,Y\}$ and $m$ is even. If $A_{j+m}=X_{j+m}$, it has the illegal sequence $X_{j+m}X_{j+m+1}$. If $A_{j+m}=Y_{j+m}$, we can swap $X_{j+m+1}Y_{j+m+2}$ to $Y_{j+m+1}X_{j+m+2}$ which also leads to an illegal sequence $Y_{j+m}Y_{j+m+1}$. Case 4 can also be handled in a similar way. For Case 3, an example is 
\begin{equation}
    (XX\;\mathcal{\overline{A}}\cdots A_{j+m}\;|\;A_{j+m+1}\cdots \mathcal{\overline{A}}\; Z)_j^k,
\end{equation}
where $A\in\{X,Y\}$ and $m$ is even. If $A_{j+m}A_{j+m+1}$ takes $X_{j+m}X_{j+m+1}$ or $Y_{j+m}Y_{j+m+1}$, this sequence is illegal. If $A_{j+m}A_{j+m+1}$ takes $X_{j+m}Y_{j+m+1}$ or $Y_{j+m}X_{j+m+1}$, it is equivalent to Case 1.

The strategy developed for even $k$ fails to determine the coefficients of charges in Case 1. In other words, we find that $k$-local charges with odd number $k$ do not connect to any charge with a zero coefficient. Our strategy is to extend the analyze used in the $5$-local case, by examining the complete set of commutators for specific operators. Inspired by the methods in \cite{shiraishi2024absenceXIX}, we construct a series of commutators and solve the corresponding coefficients relations. 

Consider the $k$-local operator 
\begin{equation}\label{selected_odd_length_charge}
    (XX\;\mathcal{\overline{A}}\;\mathcal{\overline{A}}\cdots \mathcal{\overline{A}}\;\mathcal{\overline{A}}\;Z)_j^k,
\end{equation}
where we set $\overline{\mathcal{A}} = YX$ without loss of generality. It has the following commutator
\begin{equation}\label{odd_special_charge_1} 
\begin{array}{ccccccccccc}
     &X_{j}&X&Y&X&\cdots&Y&X&Y&X&Z_{j+k-1} \\
     &&&&&&&Y&Y&&\\\hline
     &X_{j}&X&Y&X&\cdots&Y&Z&I&X&Z_{j+k-1}.
\end{array}
\end{equation}
It is connected to other charges through the commutators
\begin{subequations}
\begin{equation}\label{odd_special_charge_3} 
\begin{array}{ccccccccccc}
     &X_{j}&X&Y&X&\cdots&Y&Z&I&Y_{j+k-2}& \\
     &&&&&&&&&Z&Z\\\hline
     &X_{j}&X&Y&X&\cdots&Y&Z&I&X&Z_{j+k-1},
\end{array}
\end{equation}
\begin{equation}\label{odd_special_charge_4} 
\begin{array}{ccccccccccc}
     &&&X_{j+2}&X&\cdots&Y&Z&I&X&Z_{j+k-1} \\
     &X_{j}&X&Z&&&&&&&\\\hline
     -&X_{j}&X&Y&X&\cdots&Y&Z&I&X&Z_{j+k-1}.
\end{array} 
\end{equation}
\end{subequations}
Here we temporarily ignore the coefficient $-2$ of Hamiltonian terms $YY$ and $ZZ$, and we will take it into account when considering the coefficients relations later. Note that the charge in Eq. (\ref{odd_special_charge_3}) are $k-1$-local, and the charge in Eq. (\ref{odd_special_charge_4}) is $(k-2)$-local. 

The next task is to determine the coefficients of charges in Eqs. (\ref{odd_special_charge_3}) and (\ref{odd_special_charge_4}), therefore solving the coefficient of operator $(XX\;\mathcal{\overline{A}}\;\mathcal{\overline{A}}\cdots \mathcal{\overline{A}}\;\mathcal{\overline{A}}\;Z)_j^k$. 

Before we proceed, we define the following three notations, which also has been applied in \cite{shiraishi2024absenceXIX}, to simplify the expressions of commutators. 
\begin{itemize}
    \item The symbol ``$\underset{yy}{\up}$" denotes the commutator between an operator and $YY$ term. For example, $\cdots\mathcal{\overline{A}}_{j+2}\underset{yy}{\up}\mathcal{\overline{A}}_{j+4}\cdots$ signifies that the first ``$Y$" acts on the second operator of $\mathcal{\overline{A}}_{j+2}$ and the second ``$Y$" acts on the first operator of $\mathcal{\overline{A}}_{j+4}$.

    \item The symbol ``$\overset{yy}{|}$" represents multiplication of ``$YY$" on the left and right operators of $|$.

    \item The symbols ``$\leftarrow$" and ``$\to$" denote that commutators act at the rightmost and leftmost sites, respectively.
\end{itemize}

The above defined notations can rewrite the commutators (\ref{odd_special_charge_1}), (\ref{odd_special_charge_3}) and (\ref{odd_special_charge_4}) as
\begin{widetext}
\begin{subequations}\label{row1}
\begin{align}
\label{row1_1}
\left [ (XXYX\cdots YXYXZ)_j^k,Y_{j+k-4}Y  \right ] = X_{j}X\mathcal{\overline{A}}_{j+2}\mathcal{\overline{A}}_{j+4}\cdots \mathcal{\overline{A}}_{j+k-5}\underset{yy}{\up}\mathcal{\overline{A}}_{j+k-3}Z_{j+k-1},  
\end{align}
\begin{align}
\label{row1_2}
\left [ (XXYX\cdots YZIY)_j^{k-1}, Z_{j+k-2}Z  \right ] = X_{j}X\mathcal{\overline{A}}\cdots\mathcal{\overline{A}}_{j+k-5}\overset{yy}{|}Y_{j+k-3}Y_{j+k-2}\leftarrow Z_{j+k-2}Z,
\end{align}
\begin{align}
\label{row1_3}
\left [(XX\cdots YZIXZ)_{j+2}^{k-2}, X_jXZ \right ] = X_jXZ\to X_{j+2}X\mathcal{\overline{A}}\cdots \mathcal{\overline{A}}_{j+k-5}\overset{yy}{|}\mathcal{\overline{A}}_{j+k-3}Z_{j+k-1}.
\end{align}
\end{subequations}
\end{widetext}

Next we observe that the operator
\begin{equation}
    X_{j+2}X\mathcal{\overline{A}}\cdots \mathcal{\overline{A}}_{j+k-5}\overset{yy}{|}\mathcal{\overline{A}}_{j+k-3}\mathcal{\overline{A}}Z_{j+k+1},
\end{equation}
can only be generated from the commutators
\allowdisplaybreaks
\begin{subequations}\label{row2}
\begin{gather}
\label{row2_1}
    X_{j+2}X\mathcal{\overline{A}}\cdots \mathcal{\overline{A}}_{j+k-5}\underset{yy}{\up}\mathcal{\overline{A}}_{j+k-3}\mathcal{\overline{A}}Z_{j+k+1},\\
\label{row2_2}
    X_{j+2}X\mathcal{\overline{A}}\cdots\mathcal{\overline{A}}_{j+k-5}\overset{yy}{|}\mathcal{\overline{A}}_{j+k-3}YY_{j+k}\leftarrow Z_{j+k}Z,\\
\label{row2_3}
    X_{j+2}X\mathcal{\overline{A}}\cdots \mathcal{\overline{A}}_{j+k-5}\overset{yy}{|}\mathcal{\overline{A}}_{j+k-3}Z_{j+k-1}\leftarrow X_{j+k-1}XZ,\\
\label{row2_4}
    X_{j+2}XZ\to X_{j+4}X \mathcal{\overline{A}}\cdots \mathcal{\overline{A}}_{j+k-5}\overset{yy}{|}\mathcal{\overline{A}}_{j+k-3} \mathcal{\overline{A}}Z_{j+k+1},\\
\label{row2_additional}
   X_{j+2}X\mathcal{\overline{A}}\cdots \mathcal{\overline{A}}_{j+k-5}\overset{yy}{|}\mathcal{\overline{A}}_{j+k-3}IZ_{j+k}\leftarrow Y_{j+k-1}YZ.
\end{gather}
\end{subequations}
Here the charge in Eq. (\ref{row2_1}) is $k$-local; charges in Eqs. (\ref{row2_2}) and (\ref{row2_additional}) is $(k-1)$-local; charges in Eq. (\ref{row2_3})-(\ref{row2_4}) are $(k-2)$-local. 

We notice that the charge in (\ref{row2_additional}) forms pair with a $k$-local charge $(XXYX\cdots YZIZY)_j^k$ through the commutators
\begin{subequations}
\begin{align}
    \begin{array}{cccccccccc}
        &&&X_{j+2}&X&\cdots &\mathcal{\overline{A}}\overset{yy}{|}&\mathcal{\overline{A}}&I&Z_{j+k}\\
        &X_j&X&Z&&&&&&\\\hline
        -&X_j&X&Y&X&\cdots&\mathcal{\overline{A}}\overset{yy}{|}&\mathcal{\overline{A}}&I&Z_{j+k},
    \end{array}
\end{align}
\vspace{0.4cm}
\begin{align}
    \begin{array}{ccccccccc}
       & X_j & X& \cdots&\mathcal{\overline{A}}\overset{yy}{|}& Y&Z&Y_{j+k-1}&\\
         &&&&&&Y_{j+k-2}&Y&Z\\\hline 
          -& X_j & X &\cdots&\mathcal{\overline{A}}\overset{yy}{|}& Y&X&I&Z_{j+k}.
    \end{array}
\end{align}
\end{subequations}
This $k$-local charge has zero coefficient because of the illegal boundary sequences. Therefore the charge in Eq. (\ref{row2_additional}) has zero coefficient. With similar method, we can omit the contribution of commutators with structure similar to Eq. (\ref{row2_additional}).

We further notice that the charge in Eq. (\ref{row2_3}) is identical to the charge in Eq. (\ref{row1_3}). We can iteratively find the charges connected to the charge in Eq. (\ref{row2_4}). Eventually we obtain a series of connected charges therefore a set of equations of their coefficients. See Tables \ref{odd k sequence of commutators 1} and \ref{odd k sequence of commutators 2}.

\begin{table*}[!ht]
    \centering
    \caption{Series of commutators generated the same operators. The commutators in each row generate the same operator, and thus produces a linear relation. The Column 3rd and 4th are shown in Table \ref{odd k sequence of commutators 2}.}
    \begin{tabular}{c|c|c}
    \hline
        ~Row~ & Column 1st & Column 2nd \\ \hline
        1st & $X_{j}X\mathcal{\overline{A}}\cdots \mathcal{\overline{A}}_{j+k-5}\underset{yy}{\up}\mathcal{\overline{A}}_{j+k-3}Z_{j+k-1}$ & $X_{j}X\mathcal{\overline{A}}\cdots\mathcal{\overline{A}}_{j+k-5}\overset{yy}{|}Y_{j+k-3}Y_{j+k-2} \leftarrow Z_{j+k-2}Z$ \\ 
        2nd & $X_{j+2}X\mathcal{\overline{A}}\cdots\mathcal{\overline{A}}_{j+k-5}\underset{yy}{\up}\mathcal{\overline{A}}_{j+k-3}\mathcal{\overline{A}}Z_{j+k+1}$ & $X_{j+2}X\mathcal{\overline{A}}\cdots\mathcal{\overline{A}}_{j+k-5}\overset{yy}{|}\mathcal{\overline{A}}_{j+k-3}YY_{j+k}\leftarrow Z_{j+k}Z$ \\ 
        3rd & $X_{j+4}X\mathcal{\overline{A}}\cdots\mathcal{\overline{A}}_{j+k-5}\underset{yy}{\up}\mathcal{\overline{A}}_{j+k-3}\mathcal{\overline{A}}\mathcal{\overline{A}}Z_{j+k+3}$ & $X_{j+4}X\mathcal{\overline{A}}\cdots\mathcal{\overline{A}}_{j+k-5}\overset{yy}{|}\mathcal{\overline{A}}_{j+k-3}\mathcal{\overline{A}}YY_{j+k+2}\leftarrow Z_{j+k+2}Z$ \\ 
        $\vdots$ & $\vdots$ & $\vdots$ \\
        $\vdots$ & ~$X_{j+k-9}X\mathcal{\overline{A}}\mathcal{\overline{A}}_{j+k-5}\underset{yy}{\up}\mathcal{\overline{A}}_{j+k-3}\cdots \mathcal{\overline{A}}_{j+2k-12}Z_{j+2k-10}$~ & ~$X_{j+k-9}X\mathcal{\overline{A}}\mathcal{\overline{A}}_{j+k-5}\overset{yy}{|}\mathcal{\overline{A}}_{j+k-3}\cdots \mathcal{\overline{A}}YY_{j+2k-11}\leftarrow Z_{j+2k-11}Z$~ \\
        $\vdots$ & $X_{j+k-7}X\mathcal{\overline{A}}_{j+k-5}\underset{yy}{\up}\mathcal{\overline{A}}_{j+k-3}\cdots \mathcal{\overline{A}}_{j+2k-10}Z_{j+2k-8}$ & $X_{j+k-7}X\mathcal{\overline{A}}_{j+k-5}\overset{yy}{|}\mathcal{\overline{A}}_{j+k-3}\cdots \mathcal{\overline{A}}YY_{j+2k-9}\leftarrow Z_{j+2k-9}Z$ \\
        $\vdots$ & $X_{j+k-5}X\underset{yy}{\up}\mathcal{\overline{A}}_{j+k-3}\cdots \mathcal{\overline{A}}_{j+2k-8}Z_{j+2k-6}$ & $X_{j+k-5}X\overset{yy}{|}\mathcal{\overline{A}}_{j+k-3}\cdots \mathcal{\overline{A}}YY_{j+2k-7}\leftarrow Z_{j+2k-7}Z$ \\
        \hline
    \end{tabular}
    \label{odd k sequence of commutators 1}
\end{table*}

\begin{table*}[!ht]
    \centering
    \caption{The continued column 3rd and 4th of Table \ref{odd k sequence of commutators 1}.}
    \begin{tabular}{c|c|c}
    \hline
        Row & Column 3rd & Column 4th \\ \hline
        1st & ~ & $X_{j}XZ\to X_{j+2}X\mathcal{\overline{A}}\cdots\mathcal{\overline{A}}_{j+k-5}\overset{yy}{|}\mathcal{\overline{A}}_{j+k-3}Z_{j+k-1}$ \\ 
        2nd & $X_{j+2}X\mathcal{\overline{A}}\cdots\mathcal{\overline{A}}_{j+k-5}\overset{yy}{|}\mathcal{\overline{A}}_{j+k-3}Z_{j+k-1}\leftarrow X_{j+k-1}XZ$  & $X_{j+2}XZ\to X_{j+4}X\mathcal{\overline{A}}\cdots\mathcal{\overline{A}}_{j+k-5}\overset{yy}{|}\mathcal{\overline{A}}_{j+k-3}\mathcal{\overline{A}}Z_{j+k+1}$ \\ 
        3rd & $X_{j+4}X\mathcal{\overline{A}}\cdots\mathcal{\overline{A}}_{j+k-5}\overset{yy}{|}\mathcal{\overline{A}}_{j+k-3}\mathcal{\overline{A}}Z_{j+k+1}\leftarrow X_{j+k+1}XZ $ & $X_{j+4}XZ\to X_{j+6}X\mathcal{\overline{A}}\cdots \mathcal{\overline{A}}_{j+k-5}\overset{yy}{|}\mathcal{\overline{A}}_{j+k-3}\mathcal{\overline{A}}\mathcal{\overline{A}}Z_{j+k+3}$ \\ 
        $\vdots$ & $\vdots$ &$\vdots$ \\ 
        $\vdots$ & ~$X_{j+k-9}X\mathcal{\overline{A}}\mathcal{\overline{A}}_{j+k-5}\overset{yy}{|}\mathcal{\overline{A}}_{j+k-3}\cdots \mathcal{\overline{A}}Z_{j+2k-12}\leftarrow X_{j+2k-12}XZ$~ & ~$X_{j+k-9}XZ\to X_{j+k-7}X\mathcal{\overline{A}}_{j+k-5}\overset{yy}{|}\mathcal{\overline{A}}_{j+k-3}\cdots \mathcal{\overline{A}}Z_{j+2k-10}$~ \\
        $\vdots$ & $X_{j+k-7}X\mathcal{\overline{A}}_{j+k-5}\overset{yy}{|}\mathcal{\overline{A}}_{j+k-3}\cdots \mathcal{\overline{A}}Z_{j+2k-10}\leftarrow X_{j+2k-10}XZ$ & $X_{j+k-7}XZ\to X_{j+k-5}X\overset{yy}{|}\mathcal{\overline{A}}_{j+k-3}\cdots \mathcal{\overline{A}}Z_{j+2k-8}$ \\ 
        $\vdots$ & $X_{j+k-5}X\overset{yy}{|}\mathcal{\overline{A}}_{j+k-3}\cdots \mathcal{\overline{A}}Z_{j+2k-8}\leftarrow X_{j+2k-8}XZ$ & ~ \\
        \hline
    \end{tabular}
    \label{odd k sequence of commutators 2}
\end{table*}

Next we denote $q_{m,n}$ as the coefficient of the charge located in the $m$-th row and $n$-th column of the Tables \ref{odd k sequence of commutators 1} and \ref{odd k sequence of commutators 2}. For example,
\begin{subequations}
\begin{align}
    & q_{1,1}=q_{X_{j}X\mathcal{\overline{A}}\cdots \mathcal{\overline{A}}_{j+k-5}\mathcal{\overline{A}}_{j+k-3}Z_{j+k-1}},\\
    & q_{2,1}=q_{X_{j+2}X\mathcal{\overline{A}}\cdots\mathcal{\overline{A}}_{j+k-5}\mathcal{\overline{A}}_{j+k-3}\mathcal{\overline{A}}Z_{j+k+1}}.
\end{align}
\end{subequations}
Using the above notations, we obtain the following linear relations,
\begin{widetext}
\begin{align}
\begin{array}{ccccc}\label{coefficients relations oddk}
 -2 q_{1,1} & -2 q_{1,2} &  & -q_{1,4}  & =0,\\
    -2q_{2,1}&  -2 q_{2,2}& + q_{2,3} &  - q_{2,4}&=0, \\
 -2q_{3,1} & -2 q_{3,2}& +q_{3,3} & - q_{3,4}& =0,\\
 \vdots & \vdots & \vdots & \vdots & \vdots \\
 -2q_{(k-3)/2-2,1}& -2 q_{(k-3)/2-2,2} & +q_{(k-3)/2-2,3} & -q_{(k-3)/2-2,4}&=0,\\
 -2q_{(k-3)/2-1,1}& -2 q_{(k-3)/2-1,2} & +q_{(k-3)/2-1,3}&-q_{(k-3)/2-1,4}&=0,\\
 -2q_{(k-3)/2,1}&-2 q_{(k-3)/2,2}&+q_{(k-3)/2,3}&&=0.\\
\end{array}
\end{align}
\end{widetext}
As previously noted, $q_{1,4}$ and $q_{2,3}$ corresponds to the same charge. Similarly, $q_{2,4}$ and $q_{3,3}$ are also from the same charge. Therefore these linear equations are coupled and can be solved systematically. 

We next establish relations among charges in the first column, which corresponds to the $k$-local charges. Consider $q_{1,1}$ and $q_{2,1}$. the following commutators 
\begin{equation}
    \begin{matrix}
  &X_{j}&X&\mathcal{\overline{A}}&\cdots&\mathcal{\overline{A}}&Z_{j+k-1} &  & \\
  &  &  &  &  &    & X_{j+k-1} & X & Z_{j+k+1}\\\hline
  &X_{j}&X&\mathcal{\overline{A}}&\cdots&\mathcal{\overline{A}}&Y_{j+k-1}  &X & Z_{j+k+1},
\end{matrix}
\end{equation}
\begin{equation}
\begin{matrix}
  &  &  X_{j+2}&X&\mathcal{\overline{A}}&\cdots&\mathcal{\overline{A}}&Z_{j+k+1} \\
  X_{j}& X & Z_{j+2} & &  &  &   & \\\hline
  -X_{j}& X_{j+1}  &Y_{j+2}&X&\mathcal{\overline{A}}&\cdots&\mathcal{\overline{A}}&Z_{j+k+1},
\end{matrix}
\end{equation}
demonstrate that the charges corresponds to $q_{1,1}$ and $q_{2,1}$ are also connected. Specifically, we have $q_{1,1}=q_{2,1}$. Using similar analysis, we finally arrive at $q_{1,1}=q_{2,1}=\cdots=q_{(k-3)/2-1,1}$.

Next, we examine the coefficients in the second column, which correspond to $(k-1)$-local charges. For example, the charges of $q_{1,2}$ and $q_{2,2}$ are connected by
\begin{equation}
\begin{array}{ccccccccccc}
     &X_j&X&\mathcal{\overline{A}}&\cdots&\mathcal{\overline{A}}\overset{yy}{|}&Y&Y_{j+k-2}&&& \\
     &&&&&&-&Z_{j+k-2}&Y&Y&\\\hline
     -&X_j&X&\mathcal{\overline{A}}&\cdots&\mathcal{\overline{A}}\overset{yy}{|}&Y&X&Y&Y_{j+k},
\end{array}
\end{equation}
\begin{equation}
\begin{array}{cccccccccc}
     &&&X_{j+2}&X&\cdots&\mathcal{\overline{A}}\overset{yy}{|}&\mathcal{\overline{A}}&Y&Y_{j+k} \\
     &X_j&X&Z&&&&&&\\\hline
     -&X_j&X&Y&X&\cdots&\mathcal{\overline{A}}\overset{yy}{|}&\mathcal{\overline{A}}&Y&Y_{j+k},
\end{array}
\end{equation}
which gives $q_{1,2}=-q_{2,2}$. Iterating this procedures yield the relations $q_{1,2}=-q_{2,2}=q_{3,2}=-q_{4,2}\cdots$.

We further find that $q_{1,2}=0$, due to its corresponding charge cannot form a pair in generating the operator in the following commutator
\begin{align}
    \begin{array}{cccccccccc}
     &&& X_j&X&\cdots &Y&Z&I&Y_{j+k-2} \\
     &X_{j-2}&X&Z &&&&&&\\\hline
     -&X_{j-2}&X&Y&X&\cdots&Y&Z&I&Y_{j+k-2}.
    \end{array}
\end{align}
This means that the charges in the second column have zero coefficients. And the linear equations are simplified to
\begin{equation}\label{even rows}
 -2\left(\frac{k-3}{2}\right) q_{1,1}=0.
\end{equation}
We thus confirm that
\begin{equation}
     q_{1,1}=q_{X_{j}X\mathcal{\overline{A}}\cdots \mathcal{\overline{A}}_{j+k-5}\mathcal{\overline{A}}_{j+k-3}Z_{j+k-1}}=0,
\end{equation}
with $k\geq 4$. Since we have proved that all conserved charges with valid boundary sequences are connected (Lemma \ref{lemma:connected}), we conclude that Fredkin spin chain admits no $k$-local conserved charges for odd $k$.

\section{\label{sec:open_boundary_condition} Absence of boundary conserved charges}

In this section, we consider the boundary conserved charges of Fredkin spin chain with open boundary conditions $H_{\partial} = |\downarrow\rangle_1 \langle \downarrow| + |\uparrow\rangle_N \langle \uparrow|$. We establish the following theorem.
\begin{theorem}
    Under the open boundary condition $H_{\partial} = |\downarrow\rangle_1 \langle \downarrow| + |\uparrow\rangle_N \langle \uparrow|$, the Fredkin spin chain has no $k$-local boundary conserved charges with $k\le N/2$.
\end{theorem}
\begin{proof}
We find that the methods developed for the periodic boundary condition can be applied almost identically to the case of open boundary conditions. The only distinctions between the two cases stems from boundary effects introduced by $H_{\partial}$. Therefore our analysis focuses specifically on charges sit at boundary sites.

Under the open boundary condition, candidate boundary conserved charges take the form
\begin{equation}\label{OBCQexpression}
    Q=\sum_{l=1}^{k} \sum_{\mathbf{A}_j^l} q_{\mathbf{A}_j^l}\mathbf{A}_j^l,\quad j\in\left\{1, N-l+1\right\},
\end{equation}
where we set the operators $\mathbf{A}_j^l$ starts from the two boundaries. Since the Hamiltonian is symmetric, we only consider the left boundary with $j=1$. Extending the following results to the right boundary is straightforward, as the left and right boundary terms are both given by the $Z$ operator.

We first consider the $1$-local boundary conserved charges, which are based on the following commutators
\begin{align}
\begin{array}{r}
   X_1\\
  aZ_1 \\\hline
   -aY_1,
\end{array}\qquad
\begin{array}{r}
   Y_1\\
  aZ_1 \\\hline
   aX_1,
\end{array}\qquad
\begin{array}{rc}
 Z_1 & \\
-X_1  & X_2\\\hline
-Y_1  &X_2.
\end{array}
\end{align}
Here $a$ is a real number corresponding to coefficient of $Z_1$ in the Hamiltonian. The operators $Y_1$, $X_1$ and $Y_1X_2$ can only be generated by the above commutators, therefore can not form pair and no $1$-local boundary conserved charges. 

We next consider $k=2$ in Eq. (\ref{OBCQexpression}), which includes both $1$-local and $2$-local boundary charges. Here we take $X_1X_2$ as an example. We find that it fails to form pair in generating $Y_1 I X_3$, given by
\begin{align}
\begin{array}{cccc}
  & X_1 & X_2 & \\
 - & Z_1 & X_2 & X_3\\\hline
  & Y_1 & I_2 & X_3.
\end{array}
\end{align}
The isolated generator leads to $q_{X_1X_2}=0$. Similar analysis shows that other $2$-local boundary charges, and the $1$-local charges $X_1$ and $Y_1$ all have zero coefficients. For $Z_1$, we find that it forms pair with $X_1X_2$ through
\begin{align}
\begin{array}{cccc}
  & Z_1 &  & \\
  & Y_1 & Y_2 & Z_3\\\hline
 - & X_1 & Y_2 & Z_3,
\end{array}\qquad
\begin{array}{cccc}
  & X_1 & X_2 & \\
  & -2 & Z_2 & Z_3\\\hline
  2& X_1 & Y_2 & Z_3.
\end{array}
\end{align}
Then we obtain $q_{Z_1}=0$, since $q_{X_1X_2}=0$. This completes the proof 
that no $1$-local and $2$-local boundary conserved charge exists. 

Then we continue the analyze on $k=3$, in which includes $1,2,3$-local boundary charges simultaneously. We find that $3$-local charges can be classified into three classes based on their rightmost operator
\begin{align}\label{OBC3-local}
    A_1A_2X_3,\;
    A_1A_2Y_3,\;
    A_1A_2Z_3.
\end{align}
The commutators of the above operators with the $3$-local bulk Hamiltonian terms generate eight distinct types of $5$-local operators, given by
\begin{align}\label{OBC3to5}
\begin{array}{c}
    A_1A_2Y_3X_4X_5,\; A_1A_2Y_3Y_4Y_5,\;A_1A_2Z_3Y_4Z_5,\\
    A_1A_2X_3X_4X_5,\; A_1A_2X_3Y_4Y_5,\; A_1A_2Z_3X_4Z_5,\\
    A_1A_2Y_3X_4Z_5,\; A_1A_2X_3Y_4Z_5.
\end{array}
\end{align}
Each $5$-local operator appears only once. See Appendix \ref{App_sec:OBC_5-local}) and Tables \ref{table:boundary_3-local1}-\ref{table:boundary_3-local3}. Therefore no two $3$-local charges form pairs in generating the same 5-local operator. And we obtain $q_{\mathbf{A}_1^3}=0$ for all $\mathbf{A}_1^3$. The remaining $1$-local and $2$-local charges are also absence.

Extending the above arguments for $k$-local boundary charges is straightforward. It includes three steps.
\begin{enumerate}[(1)]
\item Consider the spin chain with $N\geq 2k$. The $(k+2)$-local boundary operators can only be generated by $k$-local boundary charges and $3$-local Hamiltonian terms.

\item Analogous to Eq. (\ref{OBC3-local}), we classify $k$-local boundary charges into three classes based on the rightmost operator.

\item All $k$-local charges cannot form pair in generating $(k+2)$-local boundary operators, and conclude that $q_{\mathbf{A}_1^k}=0$ for all $\mathbf{A}_1^k$.
\end{enumerate}

The above arguments work for the right boundary charges. We finally complete the proof that the Fredkin spin chain admits no $k$-local boundary conserved charge under the open boundary condition $H_{\partial} = |\downarrow\rangle_1 \langle \downarrow| + |\uparrow\rangle_N \langle \uparrow|$.

\end{proof}

\section{\label{sec:simplied_fredkin} Nonintegrability of simplified Fredkin spin chain}

The Motzkin spin chain, as the integer spin version of the Fredkin spin chain, posses integrable structure if certain terms in the Hamiltonian are removed \cite{Tong2020ShorMovassaghCL,Hao2022ExactSO}. Note that the Motzkin spin chain only has nearest-neighbor interactions. In this section, we examine whether similar integrable structures exist for the truncated Fredkin spin chain, in which one or several three-site interaction term is removed from the original Fredkin spin chain Hamiltonian. Our results on the truncated Fredkin spin chain can be concluded as
\begin{theorem}
    The truncated Fredkin spin chains posses local conserved charges if and only if all three-site interaction terms are removed. 
\end{theorem}

\begin{proof}
The truncation specifically targets the $3$-local Hamiltonian terms $H^{(3)}_j$ defined in Eq. (\ref{H_tilde}). We have comprehensively enumerated all possible commutators of 4-local charges in the Tables in Appendix \ref{App_sec: all_6-local}. Removing specific $3$-local Hamiltonian terms corresponds to eliminating columns in the commutators tables in Appendix \ref{App_sec: all_6-local}. We have examined that behaviors of $4$-local charges are similar to Sec.\ref{sec:absence_4} if any columns are removed. We therefore confirm that the truncated Fredkin spin chains are absence of 4-local conserved charges. 

Here, we take the example where the two terms $XXZ$ and $YYZ$ are truncated from $H^{(3)}_j$ (\ref{H_tilde}). After removing the corresponding columns in Tables \ref{table:6-local 1}-\ref{table:6-local 9}, we get fewer paired operators (blue entries). Most of the $4$-local charges now exhibit behaviors analogous to the Case 1 and Case 2 discussed in Sec. \ref{sec:absence_4}. In other words, the generated $6$-local operators appear only once or partial of them appear twice. For example, the operators $ZXYZ$ can only generate two 6-local operators, given by
\begin{align}
\begin{matrix}
  &  &  & Z & X & Y &Z \\
  -& Z & X & X &  &  & \\\hline
 - & Z & X & Y & X & Y &Z,
\end{matrix}  \qquad
\begin{matrix}
  &  &  & Z & X & Y &Z \\
  -& Z & Y & Y &  &  & \\\hline
  & Z & Y & X & X & Y &Z,
\end{matrix}
\end{align}
both of which appear one time in the simplified tables. It is worth noting that for the original Fredkin spin chain, the charge $ZXYZ$ belongs to the Case 3. 

For the above truncated Fredkin spin chain, only two $4$-local charges $ZYXX$ and $ZXYY$ form pairs in generating $6$-local operators from all terms in the truncated $H_j^{(3)}$. Applying the similar analyze used in Eqs. (\ref{4-local case3 1}) and (\ref{4-local case3 2}), we can prove that the two charges have zero coefficients. Therefore we conclude that no $4$-local conserved charges exist in such truncated Fredkin spin chain. 

For general $k$-local conserved charges, we find the truncated model imposes additional constraints on the charge sequences compared to the original model. Therefore, the absence of general $k$-local conserved charges follows the similar methods to those established in Sec. \ref{sec:absence_3_4_5} and Sec.\ref{sec:absence_general}. With the above arguments, we confirm that the truncated Fredkin spin chain is also nonintegrable.

\end{proof}

\section{\label{sec:conclusion} Conclusions and outlooks}

In this work, we have proven that the spin-$1/2$ Fredkin spin chain, originally proposed in \cite{salberger2018fredkin}, lacks $k$-local conserved charges for $4 \leq k \leq N/2$. This result confirms that the model cannot be solved using the standard algebraic Bethe Ansatz. Importantly, this conclusion holds for both periodic and open boundary conditions. Motivated by the integrable structure of the free Motzkin spin chain \cite{Tong2020ShorMovassaghCL, Hao2022ExactSO}, we have further investigated the local conserved charges of the truncated Fredkin spin chain, which is obtained by removing one or more three-site terms from the Hamiltonian. In contrast to the Motzkin spin chain, the truncated Fredkin spin chain does not exhibit any integrable structure. This suggests that integrable models with three-site interactions are inherently more intricate than those with nearest-neighbor interactions \cite{Gombor2021IntegrableSC}.

Technically, our proof extends the method proposed in \cite{shiraishi2019proofXYZh} to models that possess multiple three-site interaction terms. For general $k$-local charges, the key to the proof lies in the derivation of valid sequences, as detailed in Sec. \ref{sec:derive_sequence}. This derivation allows us to focus on a finite subset of charges, significantly simplifying the analysis. Notably, we observe that local charges with even and odd $k$ exhibit distinct behaviors. Consequently, we have employed different strategies to demonstrate that the coefficients of these charges vanish in each case.

Our work confirms the absence of local conserved charges in the Fredkin spin chain, thereby ruling out standard Yang-Baxter integrability. However, whether it possesses nontrivial quasi-local conserved charges, which are expressible as sums of densities with exponentially decaying tails and essential for characterizing nonequilibrium dynamics \cite{prosen2011open,prosen2013families,pereira2014exactly,ilievski2015quasilocal}, remains an open question. This question also lies beyond the reach of Shiraishi's method. Note that the construction of quasi-local or nonlocal conserved charges are commonly based on the transfer matrix or existing local conserved charges \cite{pereira2014exactly,ilievski2015quasilocal,bargheer2009long}. Therefore, our results also suggest that the Fredkin spin chain does not possess nontrivial quasi-local conserved charges, although inconclusive.

There are several possible extensions of our work. First, the deformed Fredkin spin chains exhibit richer structures compared to the original Fredkin spin chain \cite{salberger2017deformed, udagawa2017finite, zhang2017entropy}. It would be of significant interest to analyze the integrable structures of the deformed models and their truncated versions. Notably, the deformed Fredkin spin chain contains more three-site interaction terms than the original model. Second, recent studies have extended the analysis of conserved charges to spin-1 models \cite{hokkyo2024proof, park2025proofspin1, Shiraishi2025DichotomyTS}. This opens the possibility of determining whether the original Motzkin spin chain and other truncated versions are integrable. Third, Hokkyo recently proposed a sufficient condition for nonintegrability that does not require the analysis of $k$-local conserved charges \cite{hokkyo2025rigorous,Hokkyo2025IntegrabilityFA}. By transforming the Fredkin spin chain into a nearest-neighbor model, where each site is in the four-dimensional space, it would be intriguing to test whether the Fredkin spin chain and its truncated versions satisfy Hokkyo's condition. We leave these questions open for future investigation.

\begin{acknowledgments}

This work was supported by the NSFC (Grants No.12305028, No.12275215, No.12275214, and No.12247103), and the Youth Innovation Team of Shaanxi Universities. KZ is supported by the China Postdoctoral Science Foundation under Grant Number 2025M773421, and Scientific Research Program Funded by Education Department of Shaanxi Provincial Government (Program No.24JP186). 

\end{acknowledgments}


%

\appendix
\section{\label{App_sec: all_6-local} All 6-local operators generated by the 4-local charges}

In Tables \ref{table:6-local 1}-\ref{table:6-local 9}, we list all $6$-local operators generated by the $4$-local charges and the $3$-local Hamiltonian terms. The column corresponds to the operator generated by the same Hamiltonian term. The row corresponds to the operator generage by the same charge. Operators of the same color in the first row and column show the location of one site overlap in the commutators. For example, we have
\begin{gather}
    \qquad\begin{matrix}
     {\color{green}X} & I & I & {\color{red}X} &  & \\
      &  & - & {\color{red}Z} & X & X \\\hline
     X & I & I & Y & X & X,
   \end{matrix}\qquad
   \begin{matrix}
      &  & {\color{green}Y} & X & X & {\color{red}Y} \\
     X & X & {\color{green}Z} &  & &\\\hline
     Z & X & X & X & X & Y.
   \end{matrix}
    \end{gather}
The {\color{blue}Blue} terms represent the $6$-local operators that appear twice in Tables \ref{table:6-local 1}-\ref{table:6-local 9}.

\begin{table*}[!ht]
\caption{All $6$-local operators generated by the commutators between $4$-local charges (the first column) and $3$-local Hamiltonian terms (the first row). The green and red operators indicate the positions of operator overlap in the commutator. The blue $6$-local operators appear twice in these tables. This table only shows $4$-local charges with the structure $X\cdot\cdot X$. Commutators of other charges are shown in Tables \ref{table:6-local 2}-\ref{table:6-local 9}.}
    \centering
    \begin{tabular}{c|cccccc}
    \hline
         ~ & $-{\color{red}Z}XX$ & $-ZY{\color{green}Y}$ & $-{\color{red}Z}YY$ & $XX{\color{green}Z}$ & $YY{\color{green}Z}$ & ${\color{red}Y}YZ$ \\ \hline
        ${\color{green}X}II{\color{red}X}$ & $XIIYXX$ & $-ZYZIIX$ & $XIIYYY$ & $-XXYIIX$ & $-YYYIIX$ & $XIIZYZ$ \\ 
        ${\color{green}X}IX{\color{red}X}$ & $XIXYXX$ & $-ZYZIXX$ & $XIXYYY$ & $-XXYIXX$ & $-YYYIXX$ & $XIXZYZ$ \\ 
        ${\color{green}X}IY{\color{red}X}$ & $XIYYXX$ & $-ZYZIYX$ & $XIYYYY$ & $-XXYIYX$ & $-YYYIYX$ & $XIYZYZ$ \\ 
        ${\color{green}X}IZ{\color{red}X}$ & $XIZYXX$ & $-ZYZIZX$ & $XIZYYY$ & $-XXYIZX$ & $-YYYIZX$ & $XIZZYZ$ \\\hline 
        ${\color{green}X}XI{\color{red}X}$ & $XXIYXX$ & $-ZYZXIX$ & $XXIYYY$ & $-XXYXIX$ & $-YYYXIX$ & $XXIZYZ$ \\ 
        ${\color{green}X}XX{\color{red}X}$ & ${\color{blue}XXXYXX}$ & ${\color{blue}-ZYZXXX}$ & ${\color{blue}XXXYYY}$ & ${\color{blue}-XXYXXX}$ & ${\color{blue}-YYYXXX}$ & ${\color{blue}XXXZYZ}$ \\ 
        ${\color{green}X}XY{\color{red}X}$ & ${\color{blue}XXYYXX}$ & $-ZYZXYX$ & ${\color{blue}XXYYYY}$ & $-XXYXYX$ & $-YYYXYX$ & ${\color{blue}XXYZYZ}$ \\ 
        ${\color{green}X}XZ{\color{red}X}$ & $XXZYXX$ & $-ZYZXZX$ & $XXZYYY$ & $-XXYXZX$ & $-YYYXZX$ & $XXZZYZ$ \\ \hline
        ${\color{green}X}YI{\color{red}X}$ & $XYIYXX$ & $-ZYZYIX$ & $XYIYYY$ & $-XXYYIX$ & $-YYYYIX$ & $XYIZYZ$ \\ 
        ${\color{green}X}YX{\color{red}X}$ & $XYXYXX$ & ${\color{blue}-ZYZYXX}$ & $XYXYYY$ & ${\color{blue}-XXYYXX}$ & ${\color{blue}-YYYYXX}$ & $XYXZYZ$ \\ 
        ${\color{green}X}YY{\color{red}X}$ & $XYYYXX$ & $-ZYZYYX$ & $XYYYYY$ & $-XXYYYX$ & $-YYYYYX$ & $XYYZYZ$ \\ 
        ${\color{green}X}YZ{\color{red}X}$ & $XYZYXX$ & $-ZYZYZX$ & $XYZYYY$ & $-XXYYZX$ & $-YYYYZX$ & $XYZZYZ$ \\ \hline
        ${\color{green}X}ZI{\color{red}X}$ & $XZIYXX$ & $-ZYZZIX$ & $XZIYYY$ & $-XXYZIX$ & $-YYYZIX$ & $XZIZYZ$ \\ 
        ${\color{green}X}ZX{\color{red}X}$ & $XZXYXX$ & $-ZYZZXX$ & $XZXYYY$ & $-XXYZXX$ & $-YYYZXX$ & $XZXZYZ$ \\ 
        ${\color{green}X}ZY{\color{red}X}$ & $XZYYXX$ & $-ZYZZYX$ & $XZYYYY$ & $-XXYZYX$ & $-YYYZYX$ & $XZYZYZ$ \\ 
        ${\color{green}X}ZZ{\color{red}X}$ & $XZZYXX$ & $-ZYZZZX$ & $XZZYYY$ & $-XXYZZX$ & $-YYYZZX$ & $XZZZYZ$ \\ \hline
    \end{tabular}
    \label{table:6-local 1}
\end{table*}

\begin{table*}[!ht]
\caption{All $6$-local operators generated by the commutators between the $4$-local charges with the structure $X\cdot\cdot Y$ and the $3$-local Hamiltonian terms.}
    \centering
    \begin{tabular}{c|cccccc}
    \hline
        ~ & $-{\color{red}Z}XX$ & $-ZY{\color{green}Y}$ & $-{\color{red}Z}YY$ & $XX{\color{green}Z}$ & ${\color{red}X}XZ$ & $YY{\color{green}Z}$ \\ \hline
        ${\color{green}X}II{\color{red}Y}$ & $-XIIXXX$ & $-ZYZIIY$ & $-XIIXYY$ & $-XXYIIY$ & $-XIIZXZ$ & $-YYYIIY$ \\ 
        ${\color{green}X}IX{\color{red}Y}$ & $-XIXXXX$ & $-ZYZIXY$ & $-XIXXYY$ & $-XXYIXY$ & $-XIXZXZ$ & $-YYYIXY$ \\ 
        ${\color{green}X}IY{\color{red}Y}$ & $-XIYXXX$ & $-ZYZIYY$ & $-XIYXYY$ & $-XXYIYY$ & $-XIYZXZ$ & $-YYYIYY$ \\ 
        ${\color{green}X}IZ{\color{red}Y}$ & $-XIZXXX$ & $-ZYZIZY$ & $-XIZXYY$ & $-XXYIZY$ & $-XIZZXZ$ & $-YYYIZY$ \\ \hline
        ${\color{green}X}XI{\color{red}Y}$ & $-XXIXXX$ & $-ZYZXIY$ & $-XXIXYY$ & $-XXYXIY$ & $-XXIZXZ$ & $-YYYXIY$ \\ 
        ${\color{green}X}XX{\color{red}Y}$ & ${\color{blue}-XXXXXX}$ & $-ZYZXXY$ & ${\color{blue}-XXXXYY}$ & $-XXYXXY$ & ${\color{blue}-XXXZXZ}$ & $-YYYXXY$ \\ 
        ${\color{green}X}XY{\color{red}Y}$ & ${\color{blue}-XXYXXX}$ & ${\color{blue}-ZYZXYY}$ & ${\color{blue}-XXYXYY}$ & ${\color{blue}-XXYXYY}$ & ${\color{blue}-XXYZXZ}$ & ${\color{blue}-YYYXYY}$ \\ 
        ${\color{green}X}XZ{\color{red}Y}$ & $-XXZXXX$ & $-ZYZXZY$ & $-XXZXYY$ & $-XXYXZY$ & $-XXZZXZ$ & $-YYYXZY$ \\ \hline
        ${\color{green}X}YI{\color{red}Y}$ & $-XYIXXX$ & $-ZYZYIY$ & $-XYIXYY$ & $-XXYYIY$ & $-XYIZXZ$ & $-YYYYIY$ \\ 
        ${\color{green}X}YX{\color{red}Y}$ & $-XYXXXX$ & $-ZYZYXY$ & $-XYXXYY$ & $-XXYYXY$ & $-XYXZXZ$ & $-YYYYXY$ \\ 
        ${\color{green}X}YY{\color{red}Y}$ & $-XYYXXX$ & ${\color{blue}-ZYZYYY}$ & $-XYYXYY$ & ${\color{blue}-XXYYYY}$ & $-XYYZXZ$ & ${\color{blue}-YYYYYY}$ \\ 
        ${\color{green}X}YZ{\color{red}Y}$ & $-XYZXXX$ & $-ZYZYZY$ & $-XYZXYY$ & $-XXYYZY$ & $-XYZZXZ$ & $-YYYYZY$ \\ \hline
        ${\color{green}X}ZI{\color{red}Y}$ & $-XZIXXX$ & $-ZYZZIY$ & $-XZIXYY$ & $-XXYZIY$ & $-XZIZXZ$ & $-YYYZIY$ \\ 
        ${\color{green}X}ZX{\color{red}Y}$ & $-XZXXXX$ & $-ZYZZXY$ & $-XZXXYY$ & $-XXYZXY$ & $-XZXZXZ$ & $-YYYZXY$ \\ 
        ${\color{green}X}ZY{\color{red}Y}$ & $-XZYXXX$ & $-ZYZZYY$ & $-XZYXYY$ & $-XXYZYY$ & $-XZYZXZ$ & $-YYYZYY$ \\ 
        ${\color{green}X}ZZ{\color{red}Y}$ & $-XZZXXX$ & $-ZYZZZY$ & $-XZZXYY$ & $-XXYZZY$ & $-XZZZXZ$ & $-YYYZZY$ \\ \hline
    \end{tabular}
    \label{table:6-local 2}
\end{table*}

\begin{table*}[!ht]
\caption{All $6$-local operators generated by the commutators between the $4$-local charges with the structure $X\cdot\cdot Z$ and the $3$-local Hamiltonian terms.}
    \centering
    \begin{tabular}{c|ccccc}
    \hline
        ~ & $-ZY{\color{green}Y}$ & $XX{\color{green}Z}$ & ${\color{red}X}XZ$ & $YY{\color{green}Z}$ & ${\color{red}Y}YZ$ \\ \hline
        ${\color{green}X}II{\color{red}Z}$ & $-ZYZIIZ$ & $-XXYIIZ$ & $XIIYXZ$ & $-YYYIIZ$ & $-XIIXYZ$ \\ 
        ${\color{green}X}IX{\color{red}Z}$ & $-ZYZIXZ$ & $-XXYIXZ$ & $XIXYXZ$ & $-YYYIXZ$ & $-XIXXYZ$ \\ 
        ${\color{green}X}IY{\color{red}Z}$ & $-ZYZIYZ$ & $-XXYIYZ$ & $XIYYXZ$ & $-YYYIYZ$ & $-XIYXYZ$ \\ 
        ${\color{green}X}IZ{\color{red}Z}$ & $-ZYZIZZ$ & $-XXYIZZ$ & $XIZYXZ$ & $-YYYIZZ$ & $-XIZXYZ$ \\ \hline
        ${\color{green}X}XI{\color{red}Z}$ & $-ZYZXIZ$ & $-XXYXIZ$ & $XXIYXZ$ & $-YYYXIZ$ & $-XXIXYZ$ \\ 
        ${\color{green}X}XX{\color{red}Z}$ & $-ZYZXXZ$ & $-XXYXXZ$ & ${\color{blue}XXXYXZ}$ & $-YYYXXZ$ & ${\color{blue}-XXXXYZ}$ \\ 
        ${\color{green}X}XY{\color{red}Z}$ & ${\color{blue}-ZYZXYZ}$ & ${\color{blue}-XXYXYZ}$ & ${\color{blue}XXYYXZ}$ & ${\color{blue}-YYYXYZ}$ & ${\color{blue}-XXYXYZ}$ \\ 
        ${\color{green}X}XZ{\color{red}Z}$ & $-ZYZXZZ$ & $-XXYXZZ$ & $XXZYXZ$ & $-YYYXZZ$ & $-XXZXYZ$ \\ \hline
        ${\color{green}X}YI{\color{red}Z}$ & $-ZYZYIZ$ & $-XXYYIZ$ & $XYIYXZ$ & $-YYYYIZ$ & $-XYIXYZ$ \\ 
        ${\color{green}X}YX{\color{red}Z}$ & ${\color{blue}-ZYZYXZ}$ & ${\color{blue}-XXYYXZ}$ & $XYXYXZ$ & ${\color{blue}-YYYYXZ}$ & $-XYXXYZ$ \\ 
        ${\color{green}X}YY{\color{red}Z}$ & $-ZYZYYZ$ & $-XXYYYZ$ & $XYYYXZ$ & $-YYYYYZ$ & $-XYYXYZ$ \\ 
        ${\color{green}X}YZ{\color{red}Z}$ & $-ZYZYZZ$ & $-XXYYZZ$ & $XYZYXZ$ & $-YYYYZZ$ & $-XYZXYZ$ \\\hline 
        ${\color{green}X}ZI{\color{red}Z}$ & $-ZYZZIZ$ & $-XXYZIZ$ & $XZIYXZ$ & $-YYYZIZ$ & $-XZIXYZ$ \\ 
        ${\color{green}X}ZX{\color{red}Z}$ & ${\color{blue}-ZYZZXZ}$ & ${\color{blue}-XXYZXZ}$ & $XZXYXZ$ & ${\color{blue}-YYYZXZ}$ & $-XZXXYZ$ \\ 
        ${\color{green}X}ZY{\color{red}Z}$ & ${\color{blue}-ZYZZYZ}$ & ${\color{blue}-XXYZYZ}$ & $XZYYXZ$ & ${\color{blue}-YYYZYZ}$ & $-XZYXYZ$ \\ 
        ${\color{green}X}ZZ{\color{red}Z}$ & $-ZYZZZZ$ & $-XXYZZZ$ & $XZZYXZ$ & $-YYYZZZ$ & $-XZZXYZ$ \\ \hline
    \end{tabular}
\end{table*}

\begin{table*}[!ht]
\caption{All $6$-local operators generated by the commutators between the $4$-local charges with the structure $Y\cdot\cdot X$ and the $3$-local Hamiltonian terms.}
    \centering
    \begin{tabular}{c|cccccc}
    \hline
        ~ & $-ZX{\color{green}X}$ & $-{\color{red}Z}XX$ & $-{\color{red}Z}YY$ & $XX{\color{green}Z}$ & $YY{\color{green}Z}$ & ${\color{red}Y}YZ$ \\ \hline
        ${\color{green}Y}II{\color{red}X}$ & $ZXZIIX$ & $YIIYXX$ & $YIIYYY$ & $XXXIIX$ & $YYXIIX$ & $YIIZYZ$ \\ 
        ${\color{green}Y}IX{\color{red}X}$ & $ZXZIXX$ & $YIXYXX$ & $YIXYYY$ & $XXXIXX$ & $YYXIXX$ & $YIXZYZ$ \\ 
        ${\color{green}Y}IY{\color{red}X}$ & $ZXZIYX$ & $YIYYXX$ & $YIYYYY$ & $XXXIYX$ & $YYXIYX$ & $YIYZYZ$ \\ 
        ${\color{green}Y}IZ{\color{red}X}$ & $ZXZIZX$ & $YIZYXX$ & $YIZYYY$ & $XXXIZX$ & $YYXIZX$ & $YIZZYZ$ \\ \hline
        ${\color{green}Y}XI{\color{red}X}$ & $ZXZXIX$ & $YXIYXX$ & $YXIYYY$ & $XXXXIX$ & $YYXXIX$ & $YXIZYZ$ \\ 
        ${\color{green}Y}XX{\color{red}X}$ & ${\color{blue}ZXZXXX}$ & $YXXYXX$ & $YXXYYY$ & ${\color{blue}XXXXXX}$ & ${\color{blue}YYXXXX}$ & $YXXZYZ$ \\ 
        ${\color{green}Y}XY{\color{red}X}$ & $ZXZXYX$ & $YXYYXX$ & $YXYYYY$ & $XXXXYX$ & $YYXXYX$ & $YXYZYZ$ \\ 
        ${\color{green}Y}XZ{\color{red}X}$ & $ZXZXZX$ & $YXZYXX$ & $YXZYYY$ & $XXXXZX$ & $YYXXZX$ & $YXZZYZ$ \\\hline 
        ${\color{green}Y}YI{\color{red}X}$ & $ZXZYIX$ & $YYIYXX$ & $YYIYYY$ & $XXXYIX$ & $YYXYIX$ & $YYIZYZ$ \\ 
        ${\color{green}Y}YX{\color{red}X}$ & ${\color{blue}ZXZYXX}$ & ${\color{blue}YYXYXX}$ & ${\color{blue}YYXYYY}$ & ${\color{blue}XXXYXX}$ & ${\color{blue}YYXYXX}$ & ${\color{blue}YYXZYZ}$ \\ 
        ${\color{green}Y}YY{\color{red}X}$ & $ZXZYYX$ & ${\color{blue}YYYYXX}$ & ${\color{blue}YYYYYY}$ & $XXXYYX$ & $YYXYYX$ & ${\color{blue}YYYZYZ}$ \\ 
        ${\color{green}Y}YZ{\color{red}X}$ & $ZXZYZX$ & $YYZYXX$ & $YYZYYY$ & $XXXYZX$ & $YYXYZX$ & $YYZZYZ$ \\ \hline
        ${\color{green}Y}ZI{\color{red}X}$ & $ZXZZIX$ & $YZIYXX$ & $YZIYYY$ & $XXXZIX$ & $YYXZIX$ & $YZIZYZ$ \\ 
        ${\color{green}Y}ZX{\color{red}X}$ & $ZXZZXX$ & $YZXYXX$ & $YZXYYY$ & $XXXZXX$ & $YYXZXX$ & $YZXZYZ$ \\ 
        ${\color{green}Y}ZY{\color{red}X}$ & $ZXZZYX$ & $YZYYXX$ & $YZYYYY$ & $XXXZYX$ & $YYXZYX$ & $YZYZYZ$ \\ 
        ${\color{green}Y}ZZ{\color{red}X}$ & $ZXZZZX$ & $YZZYXX$ & $YZZYYY$ & $XXXZZX$ & $YYXZZX$ & $YZZZYZ$ \\ \hline
    \end{tabular}
\end{table*}

\begin{table*}[!ht]
\caption{All $6$-local operators generated by the commutators between the $4$-local charges with the structure $Y\cdot\cdot Y$ and the $3$-local Hamiltonian terms.}
    \centering
    \begin{tabular}{c|cccccc}
    \hline
        ~ & $-ZX{\color{green}X}$ & $-{\color{red}Z}XX$ & $-{\color{red}Z}YY$ & $XX{\color{green}Z}$ & ${\color{red}X}XZ$ & $YY{\color{green}Z}$ \\ \hline
        ${\color{green}Y}II{\color{red}Y}$ & $ZXZIIY$ & $-YIIXXX$ & $-YIIXYY$ & $XXXIIY$ & $-YIIZXZ$ & $YYXIIY$ \\ 
        ${\color{green}Y}IX{\color{red}Y}$ & $ZXZIXY$ & $-YIXXXX$ & $-YIXXYY$ & $XXXIXY$ & $-YIXZXZ$ & $YYXIXY$ \\ 
        ${\color{green}Y}IY{\color{red}Y}$ & $ZXZIYY$ & $-YIYXXX$ & $-YIYXYY$ & $XXXIYY$ & $-YIYZXZ$ & $YYXIYY$ \\ 
        ${\color{green}Y}IZ{\color{red}Y}$ & $ZXZIZY$ & $-YIZXXX$ & $-YIZXYY$ & $XXXIZY$ & $-YIZZXZ$ & $YYXIZY$ \\ \hline
        ${\color{green}Y}XI{\color{red}Y}$ & $ZXZXIY$ & $-YXIXXX$ & $-YXIXYY$ & $XXXXIY$ & $-YXIZXZ$ & $YYXXIY$ \\ 
        ${\color{green}Y}XX{\color{red}Y}$ & $ZXZXXY$ & $-YXXXXX$ & $-YXXXYY$ & $XXXXXY$ & $-YXXZXZ$ & $YYXXXY$ \\ 
        ${\color{green}Y}XY{\color{red}Y}$ & ${\color{blue}ZXZXYY}$ & $-YXYXXX$ & $-YXYXYY$ & ${\color{blue}XXXXYY}$ & $-YXYZXZ$ & ${\color{blue}YYXXYY}$ \\ 
        ${\color{green}Y}XZ{\color{red}Y}$ & $ZXZXZY$ & $-YXZXXX$ & $-YXZXYY$ & $XXXXZY$ & $-YXZZXZ$ & $YYXXZY$ \\\hline 
        ${\color{green}Y}YI{\color{red}Y}$ & $ZXZYIY$ & $-YYIXXX$ & $-YYIXYY$ & $XXXYIY$ & $-YYIZXZ$ & $YYXYIY$ \\ 
        ${\color{green}Y}YX{\color{red}Y}$ & $ZXZYXY$ & ${\color{blue}-YYXXXX}$ & ${\color{blue}-YYXXYY}$ & $XXXYXY$ & ${\color{blue}-YYXZXZ}$ & $YYXYXY$ \\ 
        ${\color{green}Y}YY{\color{red}Y}$ & ${\color{blue}ZXZYYY}$ & ${\color{blue}-YYYXXX}$ & ${\color{blue}-YYYXYY}$ & ${\color{blue}XXXYYY}$ & ${\color{blue}-YYYZXZ}$ & ${\color{blue}YYXYYY}$ \\ 
        ${\color{green}Y}YZ{\color{red}Y}$ & $ZXZYZY$ & $-YYZXXX$ & $-YYZXYY$ & $XXXYZY$ & $-YYZZXZ$ & $YYXYZY$ \\ \hline
        ${\color{green}Y}ZI{\color{red}Y}$ & $ZXZZIY$ & $-YZIXXX$ & $-YZIXYY$ & $XXXZIY$ & $-YZIZXZ$ & $YYXZIY$ \\ 
        ${\color{green}Y}ZX{\color{red}Y}$ & $ZXZZXY$ & $-YZXXXX$ & $-YZXXYY$ & $XXXZXY$ & $-YZXZXZ$ & $YYXZXY$ \\ 
        ${\color{green}Y}ZY{\color{red}Y}$ & $ZXZZYY$ & $-YZYXXX$ & $-YZYXYY$ & $XXXZYY$ & $-YZYZXZ$ & $YYXZYY$ \\ 
        ${\color{green}Y}ZZ{\color{red}Y}$ & $ZXZZZY$ & $-YZZXXX$ & $-YZZXYY$ & $XXXZZY$ & $-YZZZXZ$ & $YYXZZY$ \\ \hline
    \end{tabular}
\end{table*}

\begin{table*}[!ht]
\caption{All $6$-local operators generated by the commutators between the $4$-local charges with the structure $Y\cdot\cdot Z$ and the $3$-local Hamiltonian terms.}
    \centering
    \begin{tabular}{c|ccccc}
    \hline
        ~ & $-ZX{\color{green}X}$ & $XX{\color{green}Z}$ & ${\color{red}X}XZ$ & $YY{\color{green}Z}$ & ${\color{red}Y}YZ$ \\ \hline
        ${\color{green}Y}II{\color{red}Z}$ & $ZXZIIZ$ & $XXXIIZ$ & $-YIIYXZ$ & $YYXIIZ$ & $-YIIXYZ$ \\ 
        ${\color{green}Y}IX{\color{red}Z}$ & $ZXZIXZ$ & $XXXIXZ$ & $-YIXYXZ$ & $YYXIXZ$ & $-YIXXYZ$ \\ 
        ${\color{green}Y}IY{\color{red}Z}$ & $ZXZIYZ$ & $XXXIYZ$ & $-YIYYXZ$ & $YYXIYZ$ & $-YIYXYZ$ \\ 
        ${\color{green}Y}IZ{\color{red}Z}$ & $ZXZIZZ$ & $XXXIZZ$ & $-YIZYXZ$ & $YYXIZZ$ & $-YIZXYZ$ \\ \hline
        ${\color{green}Y}XI{\color{red}Z}$ & $ZXZXIZ$ & $XXXXIZ$ & $-YXIYXZ$ & $YYXXIZ$ & $-YXIXYZ$ \\ 
        ${\color{green}Y}XX{\color{red}Z}$ & $ZXZXXZ$ & $XXXXXZ$ & $-YXXYXZ$ & $YYXXXZ$ & $-YXXXYZ$ \\ 
        ${\color{green}Y}XY{\color{red}Z}$ & ${\color{blue}ZXZXYZ}$ & ${\color{blue}XXXXYZ}$ & $-YXYYXZ$ & ${\color{blue}YYXXYZ}$ & $-YXYXYZ$ \\ 
        ${\color{green}Y}XZ{\color{red}Z}$ & $ZXZXZZ$ & $XXXXZZ$ & $-YXZYXZ$ & $YYXXZZ$ & $-YXZXYZ$ \\ \hline
        ${\color{green}Y}YI{\color{red}Z}$ & $ZXZYIZ$ & $XXXYIZ$ & $-YYIYXZ$ & $YYXYIZ$ & $-YYIXYZ$ \\ 
        ${\color{green}Y}YX{\color{red}Z}$ & ${\color{blue}ZXZYXZ}$ & ${\color{blue}XXXYXZ}$ & ${\color{blue}-YYXYXZ}$ & ${\color{blue}YYXYXZ}$ & ${\color{blue}-YYXXYZ}$ \\ 
        ${\color{green}Y}YY{\color{red}Z}$ & $ZXZYYZ$ & $XXXYYZ$ & ${\color{blue}-YYYYXZ}$ & $YYXYYZ$ & ${\color{blue}-YYYXYZ}$ \\ 
        ${\color{green}Y}YZ{\color{red}Z}$ & $ZXZYZZ$ & $XXXYZZ$ & $-YYZYXZ$ & $YYXYZZ$ & $-YYZXYZ$ \\ \hline
        ${\color{green}Y}ZI{\color{red}Z}$ & $ZXZZIZ$ & $XXXZIZ$ & $-YZIYXZ$ & $YYXZIZ$ & $-YZIXYZ$ \\ 
        ${\color{green}Y}ZX{\color{red}Z}$ & ${\color{blue}ZXZZXZ}$ & ${\color{blue}XXXZXZ}$ & $-YZXYXZ$ & ${\color{blue}YYXZXZ}$ & $-YZXXYZ$ \\ 
        ${\color{green}Y}ZY{\color{red}Z}$ & ${\color{blue}ZXZZYZ}$ & ${\color{blue}XXXZYZ}$ & $-YZYYXZ$ & ${\color{blue}YYXZYZ}$ & $-YZYXYZ$ \\ 
        ${\color{green}Y}ZZ{\color{red}Z}$ & $ZXZZZZ$ & $XXXZZZ$ & $-YZZYXZ$ & $YYXZZZ$ & $-YZZXYZ$ \\ \hline
    \end{tabular}
\end{table*}

\begin{table*}[!ht]
\caption{All $6$-local operators generated by the commutators between the $4$-local charges with the structure $Z\cdot\cdot X$ and the $3$-local Hamiltonian terms.}
    \centering
    \begin{tabular}{c|ccccc}
    \hline
        ~ & $-ZX{\color{green}X}$ & $-{\color{red}Z}XX$ & $-ZY{\color{green}Y}$ & $-{\color{red}Z}YY$ & ${\color{red}Y}YZ$ \\ \hline
        ${\color{green}Z}II{\color{red}X}$ & $-ZXYIIX$ & $ZIIYXX$ & $ZYXIIX$ & $ZIIYYY$ & $ZIIZYZ$ \\ 
        ${\color{green}Z}IX{\color{red}X}$ & $-ZXYIXX$ & $ZIXYXX$ & $ZYXIXX$ & $ZIXYYY$ & $ZIXZYZ$ \\ 
        ${\color{green}Z}IY{\color{red}X}$ & $-ZXYIYX$ & $ZIYYXX$ & $ZYXIYX$ & $ZIYYYY$ & $ZIYZYZ$ \\ 
        ${\color{green}Z}IZ{\color{red}X}$ & $-ZXYIZX$ & $ZIZYXX$ & $ZYXIZX$ & $ZIZYYY$ & $ZIZZYZ$ \\ \hline
        ${\color{green}Z}XI{\color{red}X}$ & $-ZXYXIX$ & $ZXIYXX$ & $ZYXXIX$ & $ZXIYYY$ & $ZXIZYZ$ \\ 
        ${\color{green}Z}XX{\color{red}X}$ & ${\color{blue}-ZXYXXX}$ & $ZXXYXX$ & ${\color{blue}ZYXXXX}$ & $ZXXYYY$ & $ZXXZYZ$ \\ 
        ${\color{green}Z}XY{\color{red}X}$ & $-ZXYXYX$ & ${\color{blue}ZXYYXX}$ & $ZYXXYX$ & ${\color{blue}ZXYYYY}$ & ${\color{blue}ZXYZYZ}$ \\ 
        ${\color{green}Z}XZ{\color{red}X}$ & $-ZXYXZX$ & ${\color{blue}ZXZYXX}$ & $ZYXXZX$ & ${\color{blue}ZXZYYY}$ & ${\color{blue}ZXZZYZ}$ \\ \hline
        ${\color{green}Z}YI{\color{red}X}$ & $-ZXYYIX$ & $ZYIYXX$ & $ZYXYIX$ & $ZYIYYY$ & $ZYIZYZ$ \\ 
        ${\color{green}Z}YX{\color{red}X}$ & ${\color{blue}-ZXYYXX}$ & ${\color{blue}ZYXYXX}$ & ${\color{blue}ZYXYXX}$ & ${\color{blue}ZYXYYY}$ & ${\color{blue}ZYXZYZ}$ \\ 
        ${\color{green}Z}YY{\color{red}X}$ & $-ZXYYYX$ & $ZYYYXX$ & $ZYXYYX$ & $ZYYYYY$ & $ZYYZYZ$ \\ 
        ${\color{green}Z}YZ{\color{red}X}$ & $-ZXYYZX$ & ${\color{blue}ZYZYXX}$ & $ZYXYZX$ & ${\color{blue}ZYZYYY}$ & ${\color{blue}ZYZZYZ}$ \\ \hline
        ${\color{green}Z}ZI{\color{red}X}$ & $-ZXYZIX$ & $ZZIYXX$ & $ZYXZIX$ & $ZZIYYY$ & $ZZIZYZ$ \\ 
        ${\color{green}Z}ZX{\color{red}X}$ & $-ZXYZXX$ & $ZZXYXX$ & $ZYXZXX$ & $ZZXYYY$ & $ZZXZYZ$ \\ 
        ${\color{green}Z}ZY{\color{red}X}$ & $-ZXYZYX$ & $ZZYYXX$ & $ZYXZYX$ & $ZZYYYY$ & $ZZYZYZ$ \\ 
        ${\color{green}Z}ZZ{\color{red}X}$ & $-ZXYZZX$ & $ZZZYXX$ & $ZYXZZX$ & $ZZZYYY$ & $ZZZZYZ$ \\ \hline
    \end{tabular}
\end{table*}

\begin{table*}[!ht]
\caption{All $6$-local operators generated by the commutators between the $4$-local charges with the structure $Z\cdot\cdot Y$ and the $3$-local Hamiltonian terms.}
    \centering
    \begin{tabular}{c|ccccc}
    \hline
        ~ & $-ZX{\color{green}X}$ & $-{\color{red}Z}XX$ & $-ZY{\color{green}Y}$ & $-{\color{red}Z}YY$ & ${\color{red}X}XZ$ \\ \hline
        ${\color{green}Z}II{\color{red}Y}$ & $-ZXYIIY$ & $-ZIIXXX$ & $ZYXIIY$ & $-ZIIXYY$ & $-ZIIZXZ$ \\ 
        ${\color{green}Z}IX{\color{red}Y}$ & $-ZXYIXY$ & $-ZIXXXX$ & $ZYXIXY$ & $-ZIXXYY$ & $-ZIXZXZ$ \\ 
        ${\color{green}Z}IY{\color{red}Y}$ & $-ZXYIYY$ & $-ZIYXXX$ & $ZYXIYY$ & $-ZIYXYY$ & $-ZIYZXZ$ \\ 
        ${\color{green}Z}IZ{\color{red}Y}$ & $-ZXYIZY$ & $-ZIZXXX$ & $ZYXIZY$ & $-ZIZXYY$ & $-ZIZZXZ$ \\ \hline
        ${\color{green}Z}XI{\color{red}Y}$ & $-ZXYXIY$ & $-ZXIXXX$ & $ZYXXIY$ & $-ZXIXYY$ & $-ZXIZXZ$ \\ 
        ${\color{green}Z}XX{\color{red}Y}$ & $-ZXYXXY$ & $-ZXXXXX$ & $ZYXXXY$ & $-ZXXXYY$ & $-ZXXZXZ$ \\ 
        ${\color{green}Z}XY{\color{red}Y}$ & ${\color{blue}-ZXYXYY}$ & ${\color{blue}-ZXYXXX}$ & ${\color{blue}ZYXXYY}$ & ${\color{blue}-ZXYXYY}$ & ${\color{blue}-ZXYZXZ}$ \\ 
        ${\color{green}Z}XZ{\color{red}Y}$ & $-ZXYXZY$ & ${\color{blue}-ZXZXXX}$ & $ZYXXZY$ & ${\color{blue}-ZXZXYY}$ & ${\color{blue}-ZXZZXZ}$ \\\hline 
        ${\color{green}Z}YI{\color{red}Y}$ & $-ZXYYIY$ & $-ZYIXXX$ & $ZYXYIY$ & $-ZYIXYY$ & $-ZYIZXZ$ \\ 
        ${\color{green}Z}YX{\color{red}Y}$ & $-ZXYYXY$ & ${\color{blue}-ZYXXXX}$ & $ZYXYXY$ & ${\color{blue}-ZYXXYY}$ & ${\color{blue}-ZYXZXZ}$ \\ 
        ${\color{green}Z}YY{\color{red}Y}$ & ${\color{blue}-ZXYYYY}$ & $-ZYYXXX$ & ${\color{blue}ZYXYYY}$ & $-ZYYXYY$ & $-ZYYZXZ$ \\ 
        ${\color{green}Z}YZ{\color{red}Y}$ & $-ZXYYZY$ & ${\color{blue}-ZYZXXX}$ & $ZYXYZY$ & ${\color{blue}-ZYZXYY}$ & ${\color{blue}-ZYZZXZ}$ \\\hline 
        ${\color{green}Z}ZI{\color{red}Y}$ & $-ZXYZIY$ & $-ZZIXXX$ & $ZYXZIY$ & $-ZZIXYY$ & $-ZZIZXZ$ \\ 
        ${\color{green}Z}ZX{\color{red}Y}$ & $-ZXYZXY$ & $-ZZXXXX$ & $ZYXZXY$ & $-ZZXXYY$ & $-ZZXZXZ$ \\ 
        ${\color{green}Z}ZY{\color{red}Y}$ & $-ZXYZYY$ & $-ZZYXXX$ & $ZYXZYY$ & $-ZZYXYY$ & $-ZZYZXZ$ \\ 
        ${\color{green}Z}ZZ{\color{red}Y}$ & $-ZXYZZY$ & $-ZZZXXX$ & $ZYXZZY$ & $-ZZZXYY$ & $-ZZZZXZ$ \\ \hline
    \end{tabular}
\end{table*}

\begin{table*}[!ht]
\caption{All $6$-local operators generated by the commutators between the $4$-local charges with the structure $Z\cdot\cdot Z$ and the $3$-local Hamiltonian terms.}
    \centering
    \begin{tabular}{c|cccc}
    \hline
        ~ & $-ZX{\color{green}X}$ & $-ZY{\color{green}Y}$ & ${\color{red}X}XZ$ & ${\color{red}Y}YZ$ \\ \hline
        ${\color{green}Z}II{\color{red}Z}$ & $-ZXYIIZ$ & $ZYXIIZ$ & $ZIIYXZ$ & $-ZIIXYZ$ \\ 
        ${\color{green}Z}IX{\color{red}Z}$ & $-ZXYIXZ$ & $ZYXIXZ$ & $ZIXYXZ$ & $-ZIXXYZ$ \\ 
        ${\color{green}Z}IY{\color{red}Z}$ & $-ZXYIYZ$ & $ZYXIYZ$ & $ZIYYXZ$ & $-ZIYXYZ$ \\ 
        ${\color{green}Z}IZ{\color{red}Z}$ & $-ZXYIZZ$ & $ZYXIZZ$ & $ZIZYXZ$ & $-ZIZXYZ$ \\ \hline
        ${\color{green}Z}XI{\color{red}Z}$ & $-ZXYXIZ$ & $ZYXXIZ$ & $ZXIYXZ$ & $-ZXIXYZ$ \\ 
        ${\color{green}Z}XX{\color{red}Z}$ & $-ZXYXXZ$ & $ZYXXXZ$ & $ZXXYXZ$ & $-ZXXXYZ$ \\ 
        ${\color{green}Z}XY{\color{red}Z}$ & ${\color{blue}-ZXYXYZ}$ & ${\color{blue}ZYXXYZ}$ & ${\color{blue}ZXYYXZ}$ & ${\color{blue}-ZXYXYZ}$ \\ 
        ${\color{green}Z}XZ{\color{red}Z}$ & $-ZXYXZZ$ & $ZYXXZZ$ & ${\color{blue}ZXZYXZ}$ & ${\color{blue}-ZXZXYZ}$ \\\hline 
        ${\color{green}Z}YI{\color{red}Z}$ & $-ZXYYIZ$ & $ZYXYIZ$ & $ZYIYXZ$ & $-ZYIXYZ$ \\ 
        ${\color{green}Z}YX{\color{red}Z}$ & ${\color{blue}-ZXYYXZ}$ & ${\color{blue}ZYXYXZ}$ & ${\color{blue}ZYXYXZ}$ & ${\color{blue}-ZYXXYZ}$ \\ 
        ${\color{green}Z}YY{\color{red}Z}$ & $-ZXYYYZ$ & $ZYXYYZ$ & $ZYYYXZ$ & $-ZYYXYZ$ \\ 
        ${\color{green}Z}YZ{\color{red}Z}$ & $-ZXYYZZ$ & $ZYXYZZ$ & ${\color{blue}ZYZYXZ}$ & ${\color{blue}-ZYZXYZ}$ \\ \hline
        ${\color{green}Z}ZI{\color{red}Z}$ & $-ZXYZIZ$ & $ZYXZIZ$ & $ZZIYXZ$ & $-ZZIXYZ$ \\ 
        ${\color{green}Z}ZX{\color{red}Z}$ & ${\color{blue}-ZXYZXZ}$ & ${\color{blue}ZYXZXZ}$ & $ZZXYXZ$ & $-ZZXXYZ$ \\ 
        ${\color{green}Z}ZY{\color{red}Z}$ & ${\color{blue}-ZXYZYZ}$ & ${\color{blue}ZYXZYZ}$ & $ZZYYXZ$ & $-ZZYXYZ$ \\ 
        ${\color{green}Z}ZZ{\color{red}Z}$ & $-ZXYZZZ$ & $ZYXZZZ$ & $ZZZYXZ$ & $-ZZZXYZ$ \\ \hline
    \end{tabular}
    \label{table:6-local 9}
\end{table*}

\section{\label{App_sec:nonzero_sequence} Derivation of valid boundary sequences}

In Sec. \ref{sec:derive_sequence}, we present five representative types of charges with valid boundary sequences (to form pair with other charges). There exist other valid boundary sequences. Our analysis begins with sequences introduced in Sec. \ref{sec:absence_5}. The generalization of the first four sequences in Eq. (\ref{3operator sequence left}) are shown in Eqs. (\ref{nonzero sequence 1}) and (\ref{nonzero sequence 2}) in the main text.  We next address the remaining four sequences.

In terms of the fifth sequence in Eq. (\ref{3operator sequence left}), we consider the charge
\begin{align}
    Z_jX_{j+1}Y_{j+2}A_{j+3}A_{j+4}A_{j+5}A_{j+6}\cdots Z_{j+l-1},
\end{align}
which forms pair with
\begin{align}
    Z_{j+2}A_{j+3}A_{j+4}A_{j+5}A_{j+6}\cdots Z_{j+l-1}.
\end{align}
According to Eq. (\ref{3operator sequence left}), $A_{j+3}A_{j+4}$ must satisfy $A_{j+3}A_{j+4}\in \left\{XY, YX, XZ, YZ\right\}$. As $A_{j+3}A_{j+4}$ takes $XY$ or $YX$, the corresponding charge forms pair with 
\begin{align}
    Z_{j+4}A_{j+5}A_{j+6}\cdots Z_{j+l-1},
\end{align}
which would generate other different operators. The above charge can be handled in similar way. In other words, That is to let $A_{j+5}A_{j+6}$ also take $XY$ or $YX$. Iterating the procedure with $A_{j+m}A_{j+m+1}\in \left\{XY,YX\right\}$ for odd $m\ge 1$, we obtain the charges in Eq. (\ref{nonzero sequenceZ}).

Next, we analyze the sixth boundary sequence in Eq. (\ref{3operator sequence left}). Consider the charge
\begin{align}
\label{app:ZXZ}
    Z_jX_{j+1}Z_{j+2}A_{j+3}A_{j+4}A_{j+5}\cdots Z_{j+l-1},
\end{align}
which forms pair with 
\begin{align}
    Y_{j+2}A_{j+3}A_{j+4}A_{j+5}\cdots Z_{j+l-1}.
\end{align}
Similarly, $A_{j+3}$ should take $Y$, correspondingly $A_{j+4}$ takes $X$ or $Y$. It leads to the structure $Y_{j+2}Y_{j+3}\mathcal{\overline{A}}_{j+4}\cdots Z_{j+l-1}$. Then the operator in Eq. (\ref{app:ZXZ}) takes the form
\begin{align}
Z_jX_{j+1}Z_{j+2}Y_{j+3}\mathcal{\overline{A}}_{j+4}\mathcal{\overline{A}}_{j+6}\cdots Z_{j+l-1}.
\end{align}
Moreover, when deriving Eq. (\ref{nonzero sequenceZ}), we let $A_{j+m}A_{j+m+1}$ take $XY$ or $YX$ ($m\ge 1$ and $m$ is odd). If we set $A_{j+m}A_{j+m+1}$ to $XZ$ or $YZ$, we obtain the sequences
\begin{subequations}\label{more_valid_sequences}
\begin{align}
    Z_j\cdots X_{j+m}Z_{j+m+1}Y_{j+m+2}\mathcal{\overline{A}}_{j+m+3}\cdots Z_{j+l-1},\\
    Z_j\cdots Y_{j+m}Z_{j+m+1}X_{j+m+2}\mathcal{\overline{A}}_{j+m+3}\cdots Z_{j+l-1}.
\end{align}
\end{subequations}
The pattern is that the operator on the first position is $Z$, while one of the odd position takes $Z$. After the second $Z$ appears, we have the periodic structure $\mathcal{\overline{A}}$. Analysis on the remaining two boundary sequences in Eq. (\ref{3operator sequence left}) leads to similar results. Therefore we do not elaborate further.

\section{\label{App_sec:OBC_5-local}All 5-local boundary operators generated by 3-local boundary charges}

Tables \ref{table:boundary_3-local1}-\ref{table:boundary_3-local3} list all $5$-local boundary operators generated by commutators of $3$-local left boundary charges and $3$-local Hamiltonian terms. The red operators are always located at the third lattice site, which overlaps with the Hamiltonian term in the commutator. In these tables, no $5$-local operators appear twice, implying the absence of $3$-local boundary charges.

\begin{table*}[!ht]
\caption{All $5$-local operators generated by commutators between the $3$-local boundary charges (the first column) and the $3$-local Hamiltonian terms (the first row). The red operators are always located at the third lattice site. Therefore the commutation relations are anchored at the third lattice site. In this table, the $3$-local boundary charges have the structure $X_1\cdot\cdot$.}
    \centering
    \begin{tabular}{c|cccc}
    \hline
        ~ & $-{\color{red}Z_3}X_4X_5$ & $-{\color{red}Z_3}Y_4Y_5$ & ${\color{red}X_3}X_4Z_5$ & ${\color{red}Y_3}Y_4Z_5$ \\ \hline
       $X_1I{\color{red}X_3}$ & $XIYXX$ & $XIYYY$ & ~ & $XIZYZ$ \\ 
        $X_1I{\color{red}Y_3}$& $-XIXXX$ & $-XIXYY$ & $-XIZXZ$ & ~ \\ 
        $X_1I{\color{red}Z_3}$ & ~ & ~ & $XIYXZ$ & $-XIXYZ$ \\ 
        $X_1X_2{\color{red}X_3}$& $XXYXX$ & $XXYYY$ & ~ & $XXZYZ$ \\ 
        $X_1X_2{\color{red}Y_3}$ & $-XXXXX$ & $-XXXYY$ & $-XXZXZ$ & ~ \\ 
        $X_1X_2{\color{red}Z_3}$ & ~ & ~ & $XXYXZ$ & $-XXXYZ$ \\ 
        $X_1Y_2{\color{red}X_3}$ & $XYYXX$ & $XYYYY$ & ~ & $XYZYZ$ \\ 
        $X_1Y_2{\color{red}Y_3}$ & $-XYXXX$ & $-XYXYY$ & $-XYZXZ$ & ~ \\ 
        $X_1Y_2{\color{red}Z_3}$ & ~ & ~ & $XYYXZ$ & $-XYXYZ$ \\ 
        $X_1Z_2{\color{red}X_3}$ & $XZYXX$ & $XZYYY$ & ~ & $XZZYZ$ \\ 
        $X_1Z_2{\color{red}Y_3}$& $-XZXXX$ & $-XZXYY$ & $-XZZXZ$ & ~ \\ 
         $X_1Z_2{\color{red}Z_3}$ & ~ & ~ & $XZYXZ$ & $-XZXYZ$ \\ \hline
    \end{tabular}
    \label{table:boundary_3-local1}
\end{table*}

\begin{table*}[!ht]
\caption{All $5$-local operators generated by commutators between the $3$-local boundary charges with the structure $Y_1\cdot\cdot$ and the $3$-local Hamiltonian terms. }
    \centering
    \begin{tabular}{c|cccc}
    \hline
        ~ & $-{\color{red}Z_3}X_4X_5$ & $-{\color{red}Z_3}Y_4Y_5$ & ${\color{red}X_3}X_4Z_5$ & ${\color{red}Y_3}Y_4Z_5$ \\ \hline
       $Y_1I{\color{red}X_3}$ & $YIYXX$ & $YIYYY$ & ~ & $YIZYZ$ \\ 
        $Y_1I{\color{red}Y_3}$& $-YIXXX$ & $-YIXYY$ & $-YIZXZ$ & ~ \\ 
        $Y_1I{\color{red}Z_3}$ & ~ & ~ & $YIYXZ$ & $-YIXYZ$ \\ 
        $Y_1X_2{\color{red}X_3}$& $YXYXX$ & $YXYYY$ & ~ & $YXZYZ$ \\ 
        $Y_1X_2{\color{red}Y_3}$ & $-YXXXX$ & $-YXXYY$ & $-YXZXZ$ & ~ \\ 
        $Y_1X_2{\color{red}Z_3}$ & ~ & ~ & $YXYXZ$ & $-YXXYZ$ \\ 
        $Y_1Y_2{\color{red}X_3}$ & $YYYXX$ & $YYYYY$ & ~ & $YYZYZ$ \\ 
        $Y_1Y_2{\color{red}Y_3}$ & $-YYXXX$ & $-YYXYY$ & $-YYZXZ$ & ~ \\ 
        $Y_1Y_2{\color{red}Z_3}$ & ~ & ~ & $YYYXZ$ & $-YYXYZ$ \\ 
        $Y_1Z_2{\color{red}X_3}$ & $YZYXX$ & $YZYYY$ & ~ & $YZZYZ$ \\ 
        $Y_1Z_2{\color{red}Y_3}$& $-YZXXX$ & $-YZXYY$ & $-YZZXZ$ & ~ \\ 
         $Y_1Z_2{\color{red}Z_3}$ & ~ & ~ & $YZYXZ$ & $-YZXYZ$ \\ \hline
    \end{tabular}
    \label{table:boundary_3-local2}
\end{table*}

\begin{table*}[!ht]
\caption{All $5$-local operators generated by commutators between the $3$-local boundary charges with the structure $Z_1\cdot\cdot$ and the $3$-local Hamiltonian terms. }
    \centering
    \begin{tabular}{c|cccc}
    \hline
        ~ & $-{\color{red}Z_3}X_4X_5$ & $-{\color{red}Z_3}Y_4Y_5$ & ${\color{red}X_3}X_4Z_5$ & ${\color{red}Y_3}Y_4Z_5$ \\ \hline
       $Z_1I{\color{red}X_3}$ & $ZIYXX$ & $ZIYYY$ & ~ & $ZIZYZ$ \\ 
        $Z_1I{\color{red}Y_3}$& $-ZIXXX$ & $-ZIXYY$ & $-ZIZXZ$ & ~ \\ 
        $Z_1I{\color{red}Z_3}$ & ~ & ~ & $ZIYXZ$ & $-ZIXYZ$ \\ 
        $Z_1X_2{\color{red}X_3}$& $ZXYXX$ & $ZXYYY$ & ~ & $ZXZYZ$ \\ 
        $Z_1X_2{\color{red}Y_3}$ & $-ZXXXX$ & $-ZXXYY$ & $-ZXZXZ$ & ~ \\ 
        $Z_1X_2{\color{red}Z_3}$ & ~ & ~ & $ZXYXZ$ & $-ZXXYZ$ \\ 
        $Z_1Y_2{\color{red}X_3}$ & $ZYYXX$ & $ZYYYY$ & ~ & $ZYZYZ$ \\ 
        $Z_1Y_2{\color{red}Y_3}$ & $-ZYXXX$ & $-ZYXYY$ & $-ZYZXZ$ & ~ \\ 
        $Z_1Y_2{\color{red}Z_3}$ & ~ & ~ & $ZYYXZ$ & $-ZYXYZ$ \\ 
        $Z_1Z_2{\color{red}X_3}$ & $ZZYXX$ & $ZZYYY$ & ~ & $ZZZYZ$ \\ 
        $Z_1Z_2{\color{red}Y_3}$& $-ZZXXX$ & $-ZZXYY$ & $-ZZZXZ$ & ~ \\ 
         $Z_1Z_2{\color{red}Z_3}$ & ~ & ~ & $ZZYXZ$ & $-ZZXYZ$ \\ \hline
    \end{tabular}
    \label{table:boundary_3-local3}
\end{table*}

\end{document}